\documentclass{article}

\def\noheaderplainsetup{

\topmargin=0pt \headheight=0pt \headsep=0pt  \oddsidemargin=0pt \evensidemargin=0pt  \textheight=9.1truein \textwidth=6.5truein}   

\noheaderplainsetup

\usepackage{amsfonts}

\begin{document}

%    LOGICS:

\newcommand{\cltw}{\mbox{\bf CL12}}

%     MISC.:

\newcommand{\code}[1]{\ulcorner #1 \urcorner}
\newcommand{\pintimpl}{\mbox{\hspace{2pt}\raisebox{0.033cm}{\tiny $>$}\hspace{-0.18cm} \raisebox{-0.043cm}{\large --}\hspace{2pt}}} % parallel-recurrence-based reduction
\newcommand{\st}{\mbox{\raisebox{-0.05cm}{$\circ$}\hspace{-0.13cm}\raisebox{0.16cm}{\tiny $\mid$}\hspace{2pt}}}  
\newcommand{\sti}{\mbox{\raisebox{-0.02cm}
{\scriptsize $\circ$}\hspace{-0.121cm}\raisebox{0.08cm}{\tiny $.$}\hspace{-0.079cm}\raisebox{0.10cm}
{\tiny $.$}\hspace{-0.079cm}\raisebox{0.12cm}{\tiny $.$}\hspace{-0.085cm}\raisebox{0.14cm}
{\tiny $.$}\hspace{-0.079cm}\raisebox{0.16cm}{\tiny $.$}\hspace{1pt}}}
\newcommand{\intimpl}{\mbox{\hspace{2pt}$\circ$\hspace{-0.14cm} \raisebox{-0.043cm}{\Large --}\hspace{2pt}}} % branching-recurrence-based reduction
\newcommand{\fintimpl}{\mbox{\hspace{2pt}$\bullet$\hspace{-0.14cm} \raisebox{-0.058cm}{\Large --}\hspace{-6pt}\raisebox{0.008cm}{\scriptsize $\wr$}\hspace{-1pt}\raisebox{0.008cm}{\scriptsize $\wr$}\hspace{4pt}}} % dfb-reduction

\newcommand{\adi}{\hspace{2pt}\raisebox{0.02cm}{\mbox{\small $\sqsupset$}}\hspace{2pt}} 
\newcommand{\plus}{\mbox{\hspace{1pt}\raisebox{0.05cm}{\tiny\boldmath $+$}\hspace{1pt}}}
\newcommand{\minus}{\mbox{\hspace{1pt}\raisebox{0.05cm}{\tiny\boldmath $-$}\hspace{1pt}}}
\newcommand{\mult}{\mbox{\hspace{1pt}\raisebox{0.05cm}{\tiny\boldmath $\times$}\hspace{1pt}}}
\newcommand{\equals}{\mbox{\hspace{1pt}\raisebox{0.05cm}{\tiny\boldmath $=$}\hspace{1pt}}}
\newcommand{\notequals}{\mbox{\hspace{1pt}\raisebox{0.05cm}{\tiny\boldmath $\not=$}\hspace{1pt}}}
\newcommand{\successor}{\mbox{\hspace{1pt}\boldmath $'$}}

\newcommand{\elz}[1]{\mbox{$\parallel\hspace{-3pt} #1 \hspace{-3pt}\parallel$}} 
\newcommand{\elzi}[1]{\mbox{\scriptsize $\parallel\hspace{-3pt} #1 \hspace{-3pt}\parallel$}}
\newcommand{\emptyrun}{\langle\rangle} 
\newcommand{\oo}{\bot}            
\newcommand{\pp}{\top}            
\newcommand{\xx}{\wp}               
\newcommand{\legal}[2]{\mbox{\bf Lr}^{#1}_{#2}} 
\newcommand{\win}[2]{\mbox{\bf Wn}^{#1}_{#2}} 
\newcommand{\seq}[1]{\langle #1 \rangle}           

%     OPERATORS:

\newcommand{\pst}{\mbox{\raisebox{-0.01cm}{\scriptsize $\wedge$}\hspace{-4pt}\raisebox{0.16cm}{\tiny $\mid$}\hspace{2pt}}}
\newcommand{\pcost}{\mbox{\raisebox{0.12cm}{\scriptsize $\vee$}\hspace{-4pt}\raisebox{0.02cm}{\tiny $\mid$}\hspace{2pt}}}

\newcommand{\gneg}{\mbox{\small $\neg$}}                  %game negation
\newcommand{\mli}{\hspace{2pt}\mbox{\small $\rightarrow$}\hspace{2pt}}                      %strong reduction
\newcommand{\cla}{\mbox{$\forall$}}      %blind universal quantifier
\newcommand{\cle}{\mbox{$\exists$}}        %blind existential quantifier
\newcommand{\mld}{\hspace{2pt}\mbox{\small $\vee$}\hspace{2pt}}     %multiplicative disjunction
\newcommand{\mlc}{\hspace{2pt}\mbox{\small $\wedge$}\hspace{2pt}}   %multiplicative conjunction
\newcommand{\mlci}{\hspace{2pt}\mbox{\footnotesize $\wedge$}\hspace{2pt}}   %multiplicative conjunction
\newcommand{\ade}{\mbox{\large $\sqcup$}}      %additive existential quantifier
\newcommand{\ada}{\mbox{\large $\sqcap$}}      %additive universal quantifier
\newcommand{\add}{\hspace{2pt}\mbox{\small $\sqcup$}\hspace{2pt}}                     %additive disjunction
\newcommand{\adc}{\hspace{2pt}\mbox{\small $\sqcap$}\hspace{2pt}} 
\newcommand{\adci}{\hspace{2pt}\mbox{\footnotesize $\sqcap$}\hspace{2pt}}              %index additive conjunction
\newcommand{\clai}{\forall}     %index blind universal quantifier
\newcommand{\clei}{\exists}        %index blind existential quantifier
\newcommand{\tlg}{\bot}               %classical \bot; trivially lost elementary game
\newcommand{\twg}{\top}               %classical \top; trivially won elementary game
\newcommand{\col}[1]{\mbox{$#1$:}}

%   NUMERATED ITEMS and ENVIRONMENTS

\newtheorem{theoremm}{Theorem}[section]
\newtheorem{factt}[theoremm]{Fact}
\newtheorem{corollaryy}[theoremm]{Corollary}
\newtheorem{definitionn}[theoremm]{Definition}
\newtheorem{thesiss}[theoremm]{Thesis}
\newtheorem{lemmaa}[theoremm]{Lemma}
\newtheorem{conventionn}[theoremm]{Convention}
\newtheorem{examplee}[theoremm]{Example}
\newtheorem{exercisee}[theoremm]{Exercise}
\newtheorem{remarkk}[theoremm]{Remark}
\newenvironment{definition}{\begin{definitionn} \em}{ \end{definitionn}}
\newenvironment{thesis}{\begin{thesiss} \em}{ \end{thesiss}}
\newenvironment{theorem}{\begin{theoremm}}{\end{theoremm}}
\newenvironment{lemma}{\begin{lemmaa}}{\end{lemmaa}}
\newenvironment{fact}{\begin{factt}}{\end{factt}}
\newenvironment{corollary}{\begin{corollaryy}}{\end{corollaryy}}
\newenvironment{convention}{\begin{conventionn} \em}{\end{conventionn}}
\newenvironment{example}{\begin{examplee} \em}{\end{examplee}}
\newenvironment{exercise}{\begin{exercisee} \em}{\end{exercisee}}
\newenvironment{remark}{\begin{remarkk} \em}{\end{remarkk}}
\newenvironment{proof}{ {\bf Proof.} }{\  \rule{2.5mm}{2.5mm} \vspace{.2in} }

\title{A logical basis for constructive systems}
\author{Giorgi Japaridze\\ \\ 
{\small School of Computer Science and Technology, Shandong University, PRC;}\\
{\small Department of Computing Sciences, Villanova University, USA}}
\date{}
\maketitle

\begin{abstract} The work is devoted to {\em Computability logic} (CoL) ---  the philosophical/mathematical platform and long-term project  for  redeveloping classical logic after replacing  {\em truth} by {\em computability} in its underlying semantics. 
This article elaborates some basic complexity theory   for the CoL framework. Then it proves soundness and completeness for the deductive system {\bf CL12} with respect to the  semantics of CoL, including  the  version of the latter based on polynomial time computability instead of computability-in-principle. {\bf CL12} is a sequent calculus system, where the meaning of a sequent intuitively can be characterized as ``the succedent is algorithmically reducible to the antecedent'', and where formulas are built from predicate letters, function letters, variables, constants, identity, negation, parallel  and choice connectives,  and blind and choice quantifiers. A case is made that  {\bf CL12} is an adequate  logical basis for constructive applied theories, including  complexity-oriented ones. 
\end{abstract}


\


%\tableofcontents

\section{Introduction}\label{intr}
%\marginpar{intr}

{\bf Computability logic}, to which this contribution is devoted and which we shall henceforth refer to as {\bf CoL},   was introduced in \cite{Jap03,Japic,Japfin} as a semantically conceived open-ended framework and long-term research project for redeveloping logic as a formal theory of {\em computability}. That is as opposed to the more traditional view of logic as a formal theory of {\em truth}.  The expressions of CoL --- formulas, sequents, cirquents --- stand for interactive computational problems, understood as games played by a machine against its environment, and  computability of such problems means existence of a machine that always wins. The main ambition of the overall CoL project is to provide ever more expressive and powerful tools for systematically telling what can be computed and how,  just as classical logic is a systematic tool for telling what is true. 

Finding new converts is not among the pursuits of the present work. Numerous articles have been published on the subject in recent years (\cite{Jap03}-\cite{Ver},\cite{Xu}), and the reader is assumed to have some basic familiarity with the philosophy, motivations and techniques of CoL. If not, he or she may want to take a look at the  first 10 sections of \cite{Japfin} for a tutorial-style introduction and survey. Doing so would be helpful even if not technically necessary, as this paper provides all relevant definitions. 
 
The single deductive system dealt with in the present paper is $\cltw$. Its formulas are built in the standard way from predicate and function letters, variables, constants, identity $\equals$, negation $\gneg$, parallel connectives $\mlc,\mld,\mli$, choice connectives $\adc,\add$, blind quantifiers $\cla,\cle$ and choice quantifiers $\ada,\ade$. $\cltw$ is a sequent calculus system, where every sequent looks like 
\[E_1,\ldots,E_n\intimpl F\]
($n\geq 0$; the $E_i$ and $F$ are formulas),\footnote{The unfortunate coincidence --- or rather symmetry --- between our sequent symbol and Girard's symbol for linear implication is merely graphical. The meaning of the latter is closer to that of our $\rightarrow$.} semantically understood as an abbreviation of  
\[\st E_1\mlc \ldots\mlc \st E_n\mli F,\]
with $\st$ being the ordinary {\em branching recurrence} operator. The system is shown to be sound and complete in the sense that a sequent is $\cltw$-provable  if and only if it has a uniform (``purely logical'') solution, i.e. an algorithmic strategy that wins the game/sequent under any interpretation of its non-logical components such as predicate and function letters. Furthermore, such a strategy can be effectively --- in fact, efficiently --- extracted from a proof of the sequent. 

Logic $\cltw$  was first introduced in \cite{Japtowards}, where it was proven to be sound and complete in the above sense, but with $\fintimpl$ instead of our present $\intimpl$. The former is a version of the latter that only allows re-using antecedental resources a finite number of times (otherwise the game is considered lost even if its succedent part is won). Successfully switching from $\fintimpl$ to $\intimpl$ is a significant advance from both the philosophical and technical points of view.   As repeatedly argued in the earlier literature on CoL, it is $A\intimpl B$ --- rather than the stronger $A\fintimpl B$,    $A\pintimpl B$, $A\mli B$, etc. --- that adequately captures our ultimate, most general intuition of algorithmically reducing $B$ to $A$.  Correspondingly, it is uniform validity of 
$E_1,\ldots,E_n\intimpl F$ rather than of $E_1,\ldots,E_n\fintimpl F$ that corresponds to our ultimate intuition of ``purely logically'' reducing the succedent to the antecedent.  It can be characterized in other words as ``the succedent is a {\em logical consequence} of the antecedent''. It is exactly this intuition that is of paramount importance when dealing with $\cltw$-based applied theories such as the version {\bf CLA1} of Peano arithmetic constructed in \cite{Japtowards}. Completeness with respect to ``logical consequence'' guarantees that one can reliably use intuition on games and strategies to prove results in the theory.  Without having such a guarantee, one would be forced to resort to point-by-point syntactic derivations, which could make  successfully studying and developing $\cltw$-based applied theories next to impossible. Thus, while \cite{Japtowards} had chosen ``the right logic'' $\cltw$ as a basis for its system  of arithmetic, this choice was made just by good luck, as no  justification for it was provided or found.  

The above was about why and how the present paper strengthens the completeness result of \cite{Japtowards} for $\cltw$: $A\intimpl B$ is generally easier to win than $A\fintimpl B$, and hence the corresponding completeness theorem is harder to prove. But this  paper also strengthens --- perhaps in an even more important way --- the soundness result of \cite{Japtowards}.  For the first time in CoL's history, it brings computational complexity into the framework of the project. Namely, certain natural concepts of time and space complexities of winning strategies are defined for games. It is shown that $\cltw$ remains sound (as well as complete, of course) if one considers polynomial time  computability instead of computability-in-principle, and that the associated ``logical consequence'' relation preserves polynomial time  computability, as well as $\Omega$-time and $\Omega$-space computabilities for any class $\Omega$ of functions containing all polynomial functions and closed under composition.     This opens a whole new world of potential applications of CoL in general and $\cltw$ in particular. One can construct and explore, in a systematic way, not only computability-oriented applied theories such as the above-mentioned arithmetic {\bf CLA1}, but complexity-oriented theories as well, among the best known earlier examples of which is Buss's \cite{Buss} bounded arithmetic. For instance, a $\cltw$-based arithmetic for polynomial time computability will be an extension of classical Peano arithmetic. The single logical rule of inference of it, as well as of any other $\cltw$-based theory, would be 
\[\mbox{\em From $E_1,\ldots,E_n$ conclude $F$, as long as $\cltw$ proves } E_1,\ldots,E_n\intimpl F.\]
 Extra-Peano nonlogical axioms of such a system would be some polynomial time computable formulas/problems, such as, say, $\ada x\ade y(y=x\plus 1)$. And nonlogical rules of inference, if any, would be rules preserving the property of polynomial time computability. Then every theorem of the system, seen as an arithmetical problem,  will be  polynomial time computable, with a polynomial time solution of the problem extractable from the proof. The author expects that certain simple and elegant systems in this style can achieve not merely soundness, but   completeness as well, in the sense that every arithmetical problem with a polynomial time solution is expressed by some theorem of the system. In the same style, by varying   nonlogical extra-Peano axioms and rules, one can construct and study systems for polynomial space computability, elementary recursive computability, primitive recursive  computability, provably recursive   computability, etc. 

A notable advantage of $\cltw$-based arithmetics over the other complexity-oriented systems such as the earlier versions of bounded arithmetic and including those based on intuitionistic  logic (\cite{Bussint,Sch}), would be preserving the full expressiveness and deductive power of classical arithmetic while still being computationally and complexity-theoretically sound and meaningful. Every such system can be seen as a programming language where (efficient) programs can be automatically extracted from proofs, with ``programming'' thus simply meaning theorem-proving\footnote{In a more ambitious and, at this point, somewhat fantastic perspective, after developing reasonable theorem-provers for efficiency-oriented arithmetics, ``programming'' would simply mean stating the goal/specification --- i.e., writing a formula that represents the computational problem whose efficient solution is sought for systematic usage in the future. The compiler's job would be finding a proof (the hard part) and translating it into a solution (the easy part). The process of compiling could thus take long but, once compiled, the program would run fast ever after.} and with the (generally undecidable) problem of whether a program meets its specification being fully neutralized. Hence the importance of the just-mentioned advantage of $\cltw$-based systems over the other, known systems with similar aspirations --- which typically happen to be inherently weak theories --- is obvious: the stronger  a system, the better the chances that a proof/program will be found for a declarative, non-preprocessed,  ad hoc specification. Among the virtues of CoL is that it allows us to achieve constructive heights without throwing out the baby (such as classical logic or Peano arithmetic) with the bath water.  

The main purpose of the present publication is to provide a logical basis and   reusable point of departure for developing  complexity-oriented applied theories in the above  style --- the new line of research where the author predicts significant and fruitful activities in the near future. 

\section{Remembering constant games and some operations on them}\label{cg}
%\marginpar{cg}

Even though  the reader is expected to have some prior familiarity with CoL, for the sake of safety and convenience of references, here we reproduce the basic relevant definitions. It should also be noted that certain old concepts --- such as that of a non-constant game --- have been substantially generalized in this paper. On the other hand, certain other concepts, such as that of an HPM, have been simplified at the expense of (here unnecessary) generality. The definitions of the basic operations on games given in the present section are different from --- yet equivalent to --- the definitions of the same operations found elsewhere.   

Computational problems are understood as games between two players: $\twg$ (Machine) and $\tlg$ (Environment). 
A  {\bf move}\label{imove} means any finite string over the standard keyboard alphabet. 
A {\bf labeled move} ({\bf labmove})\label{ilabmove} is a move prefixed with $\pp$ or $\oo$, with such a prefix ({\bf label})\label{ilabel} indicating which player has made the move. 
A {\bf run}\label{irun} is a (finite or infinite) sequence of labmoves, and a {\bf position}\label{iposition} is a finite run.

We will be exclusively using the letters $\Phi,\Gamma,\Delta$  for runs, and  $\alpha,\beta$ for moves. The letter $\xx$\label{ixx} will always be a variable for players, and \[\overline{\xx}\label{ixxneg}\]  will mean ``$\xx$'s adversary'' (``the other player'').
Runs will be often delimited by ``$\langle$" and ``$\rangle$", with $\emptyrun$ thus denoting the {\bf empty run}.\label{iempty} The meaning of an expression such as $\seq{\Phi,\xx\alpha,\Gamma}$ must be clear: this is the result of appending to the position $\seq{\Phi}$ 
the labmove $\seq{\xx\alpha}$ and then the run $\seq{\Gamma}$.

The following is a formal definition of  constant games, combined with some less formal conventions regarding the usage of certain terminology.

\begin{definition}\label{game}
%\marginpar{game}
 A {\bf constant game}\label{iconstantgame} is a pair $A= (\legal{A}{},\win{A}{})$, where:

1. $\legal{A}{}$\label{ilr} is a set of runs  satisfying the condition that a (finite or infinite) run is in $\legal{A}{}$ iff all of its nonempty finite  initial
segments are in $\legal{A}{}$ (notice that this implies $\emptyrun\in\legal{A}{}$). The elements of $\legal{A}{}$ are
said to be {\bf legal runs}\label{ilegrun} of $A$, and all other runs are said to be {\bf illegal}.\label{iillegrun} We say that $\alpha$ is a {\bf legal move}\label{ilegmove} for $\xx$ in a position $\Phi$ of $A$ iff $\seq{\Phi,\xx\alpha}\in\legal{A}{}$; otherwise 
$\alpha$ is {\bf illegal}.\label{iillegmove} When the last move of the shortest illegal initial segment of $\Gamma$  is $\xx$-labeled, we say that $\Gamma$ is a {\bf $\xx$-illegal}\label{ipillegal} run of $A$. 

2. $\win{A}{}$\label{iwn}  is a function that sends every run $\Gamma$ to one of the players $\pp$ or $\oo$, satisfying the condition that if $\Gamma$ is a $\xx$-illegal run of $A$, then $\win{A}{}\seq{\Gamma}= \overline{\xx}$. When $\win{A}{}\seq{\Gamma}= \xx$, we say that $\Gamma$ is a {\bf $\xx$-won}\label{iwon} (or {\bf won by $\xx$}) run of $A$; otherwise $\Gamma$ is {\bf lost}\label{ilost} by $\xx$. Thus, an illegal run is always lost by the player who has made the first illegal move in it.  
\end{definition}

A constant game $A$ is said to be {\bf elementary} iff $\legal{A}{}=\{\emptyrun\}$, i.e., $A$ does not have any nonempty legal runs. There are exactly two elementary constant games: $\twg$ with $\win{\twg}{}\emptyrun=\pp$, and $\tlg$ with $\win{\tlg}{}\emptyrun=\oo$. Standard true sentences, such as ``snow is white'' or ``$0\equals 0$'', are understood as the game $\twg$, and false sentences, such as ``snow is black'' or ``$0\equals 1$'', as the game $\tlg$. Correspondingly, the two games $\twg$ and $\tlg$ will be referred to as {\bf propositions}.

Let us remember the operation of {\em prefixation}.\label{iprefixation}
It takes two arguments: a constant game $A$ and a position $\Phi$ 
 that must 
be a legal position of $A$ (otherwise the operation is undefined), and returns the game $\seq{\Phi}A$.
Intuitively, $\seq{\Phi}A$ is the game playing which means playing $A$ starting (continuing) from position $\Phi$. 
That is, $\seq{\Phi}A$ is the game to which $A$ {\bf evolves} (will be ``{\bf brought down}") after the moves of $\Phi$ have been made.  Here is a definition:

\begin{definition}\label{prfx}
%\marginpar{prfx}
Let $A$ be a constant game and $\Phi$ a legal position of $A$. The game 
$\seq{\Phi}A$\label{ipr} is defined by: 
\begin{itemize}
\item $\legal{\seq{\Phi}A}{}= \{\Gamma\ |\ \seq{\Phi,\Gamma}\in\legal{A}{}\}$;
\item $\win{\seq{\Phi}A}{}\seq{\Gamma}= \win{A}{}\seq{\Phi,\Gamma}$.
\end{itemize}
\end{definition}

\begin{convention}\label{poscon}
%\marginpar{poscon}
A terminological convention important to remember is that we often identify a legal position $\Phi$ of a game $A$ with the game $\seq{\Phi}A$. So, for instance, we may say that the move $1$ by $\oo$ brings the game $B_0\adc B_1$ down to the position $B_1$. Strictly speaking, $B_1$ is not a position but a game, and what {\em is} a position is $\seq{\oo 1}$, which we here identified with the game $B_1=\seq{\oo 1}(B_0\adc B_1)$.
\end{convention}

Note that, in order to define the $\legal{}{}$ component of   a constant game $A$,  it would suffice  to specify what the {\bf initial legal (lab)moves}\label{iilm} --- i.e., the elements of  $\{\xx\alpha\ |\ \seq{\xx\alpha}\in\legal{A}{}\}$ --- are, and to 
what game the game $A$ is  brought down after such an initial legal labmove $\xx\alpha$ is made (this can be seen to hold even in recursive definitions of game operations as in clauses 1 and 3 of Definition \ref{op} below).  
Then, the set of legal runs of $A$ will be uniquely defined. Similarly, note that defining the $\win{}{}$ component for only legal runs of $A$ would be sufficient, for then it uniquely extends to all runs. With these observations in mind, we can (re)define the operations $\gneg,\mlc,\mld,\adc,\add$  in a new fashion as follows:

\begin{definition}\label{op} Let $A$, $B$, $A_0,A_1,\ldots$ be   constant games, and $n$  a positive integer.\vspace{9pt}
%\marginpar{op}

\noindent 1. $\gneg A$\label{igneg2} ({\bf negation}) is defined by: 
\begin{quote}\begin{description}
\item[(i)] $\seq{\xx\alpha}\in\legal{\gneg A}{}$ iff $\seq{\overline{\xx}\alpha}\in\legal{A}{}$. Such an initial legal labmove $\xx\alpha$ brings the game down to 
$\gneg \seq{\overline{\xx}\alpha}A$.
\item[(ii)] Whenever $\Gamma$ is a legal run of $\gneg A$, $\win{\gneg A}{}\seq{\Gamma} = \pp$ iff $\win{A}{}\seq{\overline{\Gamma}} =\oo$. Here $\overline{\Gamma}$ means $\Gamma$ with each label $\xx$ changed to $\overline{\xx}$.\vspace{5pt}
\end{description}\end{quote}

\noindent 2. $A_0\adc\ldots\adc  A_n$\label{iadc2} ({\bf choice conjunction}) is defined by: 
\begin{quote}\begin{description}
\item[(i)] $\seq{\xx\alpha}\in\legal{A_0\adci\ldots\adci  A_n}{}$ iff $\xx= \oo$ and $\alpha= i  \in\{0,\ldots,n\}$.\footnote{Here the number $i$ is identified with its binary representation. The same applies to the other clauses of this definition.}  Such an initial legal labmove $\oo i$ brings the game down to 
$A_i$.
\item[(ii)] Whenever $\Gamma$ is a legal run of $A_0\adc\ldots\adc  A_n$, $\win{A_0\adci\ldots\adci  A_n}{}\seq{\Gamma} = \oo$ iff $\Gamma$ looks like $\seq{\oo i,\Delta}$ ($i  \in\{0,\ldots,n\}$) and $\win{A_i}{}\seq{\Delta} = \oo$.\vspace{5pt} 
\end{description}\end{quote}

\noindent 3. $A_0\mlc\ldots\mlc A_n$\label{imlc2} ({\bf parallel conjunction}) is defined by: 
\begin{quote}\begin{description}
\item[(i)] $\seq{\xx\alpha}\in\legal{A_0\mlci\ldots\mlci  A_n}{}$ iff $\alpha= i.\beta$, where $i\in\{0,\ldots,n\}$ and $\seq{\xx\beta}\in\legal{A_i}{}$. Such an initial legal labmove $\xx i.\beta$ brings the game down to  
\[ A_0\mlc\ldots\mlc A_{i-1}\mlc \seq{\xx\beta}A_i\mlc A_{i\plus 1}\mlc\ldots\mlc A_n.\] 
\item[(ii)] Whenever $\Gamma$ is a legal run of $A_0\mlc\ldots\mlc A_n$, $\win{A_0\mlci\ldots\mlci  A_n}{}\seq{\Gamma}= \pp$ iff, for each $i\in\{0,\ldots,n\}$,  $\win{A_i}{}\seq{\Gamma^{i.}}= \pp$. Here $\Gamma^{i.}$ means the result of removing, from $\Gamma$, all labmoves except those that look like $\xx i.\alpha$, and then further changing each such (remaining) $\xx i.\alpha$ to $\xx \alpha$.\footnote{Intuitively, $\Gamma^{i.}$ is the run played in the $A_i$ component. The present condition thus means that $\pp$ wins a $\mlci$-conjunction of games iff it wins in each conjunct.} \vspace{5pt}  
\end{description}\end{quote}

\noindent 4. $A_0\add\ldots\add A_n$\label{iadd2} ({\bf choice disjunction}) and $A_0\mld\ldots\mld  A_n$ ({\bf parallel disjunction}) 
are defined exactly as $A_0\adc\ldots\adc A_n$ and $A_0\mlc\ldots\mlc  A_n$, respectively, only with ``$\pp$" and ``$\oo$" interchanged.\vspace{7pt}

\noindent 5. The infinite $\adc$-conjunction $A_0\adc A_1\adc\ldots$ is defined exactly as $A_0\adc\ldots\adc A_n$, only with ``$i\in\{0,1,\ldots\}$" instead of ``$i\in\{0,\ldots,n\}$". Similarly for the infinite versions of $\add$, $\mlc$ and $\mld$.\vspace{7pt}

\noindent 6. $A\mli B$\label{imli2} ({\bf strict reduction}) is treated as an abbreviation of $(\gneg A)\mld B$.
\end{definition}

We also agree that, when $k=1$,   $A_1\adc\ldots\adc A_k$ simply means $A_1$, and so do $A_1\add\ldots\add A_k$, $A_1\mlc\ldots\mlc A_k$ and $A_1\mld\ldots\mld A_k$. We further agree that, when the set $\{A_1,\ldots , A_k\}$ is empty ($k=0$, that is), both  $A_1\adc\ldots\adc A_k$ and  $A_1\mlc\ldots\mlc A_k$ mean $\pp$, while both  $A_1\add\ldots\add A_k$ and  $A_1\mld\ldots\mld A_k$ mean $\oo$.

\begin{example}
 The game $(0\equals 0\adc 0\equals 1)\mli(10\equals 11\adc 10\equals 10)$, i.e. \(\gneg (0\equals 0\adc 0\equals 1)\mld(10\equals 11\adc 10\equals 10),\)
 has thirteen legal runs, which are: 
\begin{description}
\item[1] $\seq{}$. It is won by $\pp$, because $\pp$ is the winner in the right $\mld$-disjunct (consequent).
\item[2] $\seq{\pp 0.0}$. (The labmove of) this run brings the game down to $\gneg 0\equals 0\mld(10\equals 11\adc 10\equals 10)$, and $\pp$ is the winner for the same reason as in the previous case.
\item[3] $\seq{\pp 0.1}$. It brings the game down to $\gneg 0\equals 1\mld(10\equals 11\adc 10\equals 10)$, and $\pp$ is the winner because it wins in both $\mld$-disjuncts. 
\item[4] $\seq{\oo 1.0}$. It brings the game down to $\gneg(0\equals 0\adc 0\equals 1)\mld 10\equals 11$.  $\pp$ loses as it loses in both $\mld$-disjuncts. 
\item[5] $\seq{\oo 1.1}$. It brings the game down to $\gneg (0\equals 0\adc 0\equals 1)\mld 10\equals 10$.  $\pp$ wins as it wins in the right $\mld$-disjunct. 
\item[6-7] $\seq{\pp 0.0,\oo 1.0}$ and $\seq{\oo 1.0, \pp 0.0}$. Both bring the game down to the false $\gneg 0\equals 0 \mld 10\equals 11$, and both are lost by  $\pp$. 
\item[8-9] $\seq{\pp 0.1,\oo 1.0}$ and $\seq{\oo 1.0, \pp 0.1}$. Both bring the game down to the true $\gneg 0\equals 1 \mld 10\equals 11$, which makes  $\pp$ the winner.
\item[10-11] $\seq{\pp 0.0,\oo 1.1}$ and $\seq{\oo 1.1, \pp 0.0}$. Both bring the game down to the true $\gneg 0\equals 0 \mld 10\equals 10$, so $\pp$ wins.
\item[12-13] $\seq{\pp 0.1,\oo 1.1}$ and $\seq{\oo 1.1, \pp 0.1}$. Both bring the game down to the true $\gneg 0\equals 1 \mld 10\equals 10$, so $\pp$ wins.
\end{description}
\end{example}

Later we will be using some relaxed informal jargon already established in CoL for describing runs and strategies, referring to moves by their intuitive meanings or their effects on the game. For instance, the initial labmove $\pp 0.0$ in a play of the game $p\adc q\mli r$ we can characterize as ``$\pp$ made the move $0.0$''. Remembering the meaning of the prefix ``$0.$'' of this move, we may as well say that ``$\pp$ made the move $0$ in the antecedent''. Further remembering the effect of such a move on the antecedent, we may just as well say ``$\pp$ chose $p$ (or the left $\adc$-conjunct) in the antecedent''. We may also 
 say something like ``$\pp$ (made the move that) brought the game down to $p\mli r$'', or ``$\pp$ (made the move that) brought the antecedent down to $p$''. 

To (re)define the   operation $\st$ in the style of Definition \ref{op}, we need some preliminaries.  What we call a  {\bf tree of games} is a structure defined inductively as an element of the smallest set satisfying the following conditions:
\begin{itemize}
\item Every constant game $A$ is a tree of games. The one-element sequence $\seq{A}$ is said to be the {\bf yield} of such a tree, and the {\bf address} of $A$ in this tree is the empty bit string. 
\item Whenever $\cal A$ is a tree of games with yield $\seq{A_1,\ldots,A_m}$ and $\cal B$ is a  tree of games with yield $\seq{B_1,\ldots,B_n}$, the pair ${\cal A}\circ{\cal B}$ is a tree of games with yield $\seq{A_1,\ldots,A_m,B_1,\ldots,B_n}$. The {\bf address} of each $A_i$  in this tree is $0w$, where $w$ is the address of $A_i$ in $\cal A$. Similarly,  the address of each $B_i$ is $1w$, where $w$ is the address of $B_i$ in $\cal B$.
\end{itemize}

Example: Where $A,B,C,D$ are constant games, $(A\circ B)\circ(C\circ(A\circ D))$ is a tree of games with yield $\seq{A,B,C,A,D}$. The address of the first $A$ of the yield, to which we may as well refer as the first {\bf leaf} of the tree, is $00$; the address of the second leaf $B$ is $01$; the address of the third leaf $C$ is $10$; the address of the fourth leaf $A$ is $110$; and the address of the fifth leaf $D$ is $111$.

Note that $\circ$ is not an operation on games, but just a symbol used instead of the more common comma to separate the two parts of a pair. And a tree of games itself is not a game, but a collection of games arranged into a certain structure, just as a sequence of games is not a game but a collection of games arranged as a list. 

For bit strings $u$ and $w$, we will write $u\preceq w$ to indicate that $u$ is a (not necessarily proper) {\bf prefix} (initial segment) of $w$.

\begin{definition}\label{op1} 
%\marginpar{op1}
Let  $A_1,\ldots,A_n$ ($n\geq 1$) be   constant games, and $\cal T$ be a tree of games with yield $\seq{A_1,\ldots,A_n}$. Let $w_1,\ldots,w_n$ be the addresses of $A_1,\ldots,A_n$ in $\cal T$, respectively. The game   $\st {\cal T}$ (the {\bf branching recurrence} of $\cal T$) is defined by:
\begin{description}
\item[(i)] $\seq{\xx\alpha}\in\legal{\sti {\cal T}}{}$ iff one of the following conditions is satisfied: 
\begin{enumerate}
\item $\xx\alpha=\xx u.\beta$, where $u\preceq w_i$ for at least one $i\in\{1,\ldots,n\}$ and, for each  $i$ with $u\preceq w_i$,  $\seq{\xx\beta}\in\legal{A_i}{}$. We call such a move a {\bf nonreplicative} (lab)move. It brings the game down to  \(\st {\cal T}'\), where ${\cal T}'$ is the result of replacing $A_i$ by $\seq{ \xx\beta}A_i$ in $\cal T$ for each $i$ with $u\preceq w_i$. If here $u$ is $w_i$ (rather than a proper prefix of such) for one of $i\in\{1,\ldots,n\}$, we say that the  move $\xx u.\beta$ is {\bf focused}. Otherwise it is {\bf unfocused}. 
\item $\xx\alpha=\oo \col{w_i}$, where $i\in\{1,\ldots,n\}$. We call such a move a {\bf replicative} (lab)move. It brings the game down to  $\st {\cal T}'$, where ${\cal T}'$ is the result of replacing $A_i$ by $(A_i\circ A_i)$ in $\cal T$.
\end{enumerate}
\item[(ii)] Whenever $\Gamma$ is a   legal run of $\st {\cal T}$,   $\win{\sti {\cal T}}{}\seq{\Gamma}=\pp$ iff, for each $i\in\{1,\ldots,n\}$ and every infinite bit string $v$ with $w_i\preceq v$, we have $\win{A_i}{}\seq{\Gamma^{\preceq v}}=\pp$. Here $\Gamma^{\preceq v}$ means the result of deleting, from $\Gamma$, all labmoves except those that look like $\xx u.\alpha$ for some bit string $u$ with $u\preceq v$, and then further changing each such (remaining) labmove  $\xx u.\alpha$ to $\xx \alpha$.\footnote{Intuitively, $\Gamma^{\preceq v}$ is the run played in one of the multiple ``copies'' of $A_i$ that have been generated in the play, with $v$ acting as a (perhaps longer than necessary yet meaningful) ``address'' of that copy.}
\end{description}
\end{definition}

\begin{example} \label{ex37}
%\marginpar{ex37} 
 Let \[G\ =\    p\add(q\adc (r\adc (s\add t))) ,\] where $p,q,r,s,t $ are constant elementary games. And let  
\[\Gamma\ = \ \seq{ \oo \col{},\ \pp .1, \ \oo 0.0,\ \oo 1.1,\  \oo \col{1},\ \oo 10.0,\ \oo 11.1,\ \pp 11.0}.\]  Then $\Gamma$ is a legal run of $\st G$. Below we trace, step by step, the effects of its moves on $\st G$.

The 1st (lab)move $\oo \col{}$ means that $\oo$ replicates the (only) leaf of the tree, with the address of that leaf being the empty bit string. This move 
 brings the game down to --- in the sense that $\seq{\oo \col{}}\st G$ is --- the following game:
\[\st \Bigl(\bigl(p\add(q\adc (r\adc (s\add t)))\bigr)\circ \bigl(p\add(q\adc (r\adc (s\add t)))\bigr)\Bigr).\]

The 2nd move $\pp .1$ means choosing the second $\add$-disjunct $q\adc (r\adc (s\add t))$ in {\em both} leaves of the tree. This is so because the addresses of those leaves are $0$ and $1$, and the empty bit string --- seen between ``$\pp$'' and ``$.1$'' in $\pp .1$ --- is an initial segment of both addresses. The effect of this unfocused move is the same as the effect of the two consecutive focused moves $\pp 0.1$ and $\pp 1.1$ (in whatever order) would be, but $\pp$ might have its reasons for having made an unfocused move. Among such reasons could be that $\pp$ did not notice $\oo$'s initial move (or the latter arrived late over the asynchronous network) and thought that the position was still $G$, in which case making the moves $\pp 0.1$ and $\pp 1.1$ would be simply illegal. Note also that the ultimate effect of the move $\pp .1$ on the game would remain the same as it is now even if this move was made before the replicative move $\oo\col{}$. It is CoL's striving to achieve this sort of flexibility and asynchronous-communication-friendliness  that has determined our seemingly ``strange'' choice of trees rather than sequences as the underlying structures for $\st$-games. Any attempt to deal with sequences instead of trees would encounter the problem of violating what CoL calls the  {\em static} (speed-independent) property of games.\label{foot}\footnote{While not technically necessary for the purposes of this paper, here we still reproduce a definition of the static property for constant games. 
For either player $\xx$, let us say that a run $\Upsilon$ is a {\bf $\xx$-delay} of a run $\Gamma$ iff (1) 
 for both players $\xx'\in\{\top,\bot\}$, the subsequence of $\wp'$-labeled moves of $\Upsilon$ is the same as that of $\Gamma$, and (2)
for any $n,k\geq 1$, if the $n$th $\wp$-labeled move is made later than (is to the right of) the $k$th $\gneg\wp$-labeled move in $\Gamma$, then so is it in $\Upsilon$.  Next, 
let us say that a run is {\bf $\wp$-legal} iff it is not $\wp$-illegal. 
Now, we say that a constant game  $A$ is {\bf static} iff, whenever a run $\Upsilon$ is a $\wp$-delay of 
a run $\Gamma$, we have: (i) if $\Gamma$ is a $\wp$-legal run of $A$, then so is $\Upsilon$, and 
(ii) if $\Gamma$ is a $\wp$-won run of $A$, then so is $\Upsilon$.}

 Anyway, the position resulting from the second move of $\Gamma$ is    
\[\st \Bigl(\bigl( q\adc (r\adc (s\add t))\bigr)\circ \bigl(q\adc (r\adc (s\add t))\bigr)\Bigr).\]

The effect of the 3rd move $ \oo 0.0$ is choosing the left $\adc$-conjunct $q$ in the left ($0$-addressed) leaf of the tree, which results in 
\[\st \Bigl(q\circ \bigl(q\adc (r\adc (s\add t))\bigr)\Bigr).\]

Similarly, the 4th move $\oo 1.1$ chooses the right $\adc$-conjunct in the right leaf of the tree, resulting in 
\[\st \Bigl(q\circ \bigl(r\adc (s\add t)\bigr)\Bigr).\]

The 5th move $\oo \col{1}$ replicates  the right leaf, bringing the game down to  
\[\st \Bigl(q\circ \bigl((r\adc (s\add t))\circ (r\adc (s\add t))\bigr)\Bigr).\]

The 6th move $\oo 10.0$ chooses the left $\adc$-conjunct in the second ($00$-addressed) leaf, and, similarly, the 7th move $\oo 11.1$ chooses the right $\adc$-conjunct in the third ($11$-addressed) leaf. These two moves bring the game down to   
\[\st \Bigl(q\circ \bigl(r\circ (s\add t)\bigr)\Bigr).\]

The last, 8th move $\pp 11.0$ chooses the left $\add$-disjunct of the third leaf, and the final position is 
\[\st \bigl(q\circ (r\circ s)\bigr).\]

According to clause (ii) of Definition \ref{op1}, $\Gamma$ is a $\pp$-won run of $\st G$ iff, for any infinite bit string $v$, $\Gamma^{\preceq v}$ is a $\pp$-won run of $G$. Observe that for any infinite --- or, ``sufficiently long'' finite --- bit string $v$, $\Gamma^{\preceq v}$ is either $\seq{\pp 1,\oo 0}$ (if $v=0\ldots$) or $\seq{\pp 1,\oo 1, \oo 0}$ (if $v=10\ldots$) or $\seq{\pp 1,\oo 1, \oo 1, \pp 0}$ (if $v=11\ldots$).  We also have $\seq{\pp 1,\oo 0}G=q$,
$\seq{\pp 1,\oo 1, \oo 0}G=r$ and  $\seq{\pp 1,\oo 1, \oo 1, \pp 0}G=s$. So it is no accident that we see $q,r,s$ at the leaves in the final position. Correspondingly, the game is won iff each one of these three propositions is true. 

The cases where $\oo$ makes infinitely many replications in a run $\Gamma$ of a game $\st H$ and hence the ``eventual tree'' is infinite are similar, with the only difference that the ``addresses'' of the ``leaves'' of such a ``tree'', corresponding to different plays of  $H$, may be infinite bit strings.  But, again, the overall game $\st H$ will be won by $\pp$ iff all of those plays --- all $\Gamma^{\preceq v}$ where $v$ is an infinite bit string, that is --- are $\pp$-won plays of $H$.  

\end{example}

\begin{definition}\label{opp}
%\marginpar{opp}
 Let $B$,  $A_1,\ldots,A_n$ ($n\geq 0$) be   constant games. We define 
\(A_1,\ldots,A_n\intimpl B\) --- let us call it the {\bf ultimate reduction} of $B$ to $A_1,\ldots,A_n$ --- as the game \(\st A_1\mlc\ldots\mlc\st A_n\mli B.\)
\end{definition}

\section{Generalized universes and non-constant games}\label{nncg}
%\marginpar{nncg}

Constant games can be seen as generalized propositions: while propositions in classical logic are just elements 
of $\{\twg,\tlg\}$, constant games are functions from runs to $\{\twg,\tlg\}$.  Our concept of a  game generalizes that of a constant game in the same sense as the classical concept of a predicate generalizes that of a proposition.

We fix some infinite set 
of expressions called {\bf variables}, and use \label{ivariable} 
 the letters \(x,y,z,s,r,t\)  as metavariables for them. 
We also fix another infinite set (disjoint from the previous one) of expressions called {\bf constants}:\label{iconstant} 
\[\{0,1,10,11,100,101,110,111,1000,\ldots\}.\] These are thus  {\bf binary numerals}\label{ibinnum} --- the strings matching the regular expression $0\cup 1(0\cup 1)^*$.  We will be typically identifying such strings --- by some rather innocent abuse of concepts --- with the natural numbers represented by them in the standard binary notation, and vice versa.  We will be mostly using $a,b,c,d$ as metavariables for constants. 

A {\bf universe} (of discourse) is a pair $(U, ^U)$, where $U$ is a nonempty set, and $^U$, called the {\bf naming function} of the universe,  is a function that sends each constant $c$  to an element $c^U$ of $U$. The intuitive meaning of $c^U=s$ is that $c$ is a {\bf name} of $s$. Both terminologically and notationally, we will typically identify each universe $(U, ^U)$ with its first component and, instead of ``$(U, ^U)$'', write simply ``$U$'', keeping in mind that each such ``universe'' $U$ comes with a fixed  associated function $^U$. 

A universe $U=(U,^U)$ is said to be {\bf ideal} iff $U$ coincides with the above-fixed set of constants, and $^U$ is the identity function on that set.  
All earlier papers on CoL dealt only with ideal universes. This was for simplicity considerations, yielding no loss of generality as so far no results have relied on the assumption that the underlying universes were ideal. Our present treatment, however, for both technical and philosophical reasons, {\em does} call for the above-defined, more general, concept of a universe --- a universe where some objects may have  unique names, some objects have many names, and some objects have no names at all.\footnote{Further generalizations are possible if and when a need arises. Namely, one may depart from our present assumption that the set of constants is infinite and/or fixed, as long as there is a fixed constant --- say, $0$ --- that belongs to every possible set of constants ever considered. No results of this or any earlier papers on CoL would be in any way affected by doing so.} Note that real-world universes are typically not ideal: not all people living or staying in the United States have Social Security numbers;  most stars and planets of the Galaxy have no names at all, while some  have several names (Morning Star = Evening Star = Venus); etc. A natural example of a non-ideal universe from the world of mathematics would be the set  of real numbers, only some of whose elements have names, such as $5$, $1/3$, $\sqrt{2}$ or $\pi$. Generally, no uncountable universe would  be ideal  for the simple reason that  there can  only be countably many names. This is so because names, by their very nature and purpose, have to be finite objects. Observe also that many  properties of common interest such as computability or decidability, are usually sensitive to how objects are named. For instance, strictly speaking, computing a function $f(x)$ means the ability to tell, after seeing a (the) {\em name} of an arbitrary object $\mathfrak{a}$, to produce a (the) {\em name} of the object $\mathfrak{b}$ with $\mathfrak{b}=f(\mathfrak{a})$.  Similarly, an algorithm that decides a predicate $p(x)$ on a set $S$, strictly speaking, takes not {\em elements} of $S$ --- which may be abstract objects such as numbers or graphs --- but rather {\em names} of those elements (such as binary numerals or codes). It is not hard to come up with a nonstandard naming of natural numbers through binary numerals where the predicate ``$x$ is even'' is undecidable.  On the other hand, for any undecidable arithmetical predicate $p(x)$, one can come up with a naming function such that $p(x)$ becomes decidable --- for instance, one that assigns even-length names to all $\mathfrak{a}$ with $p(\mathfrak{a})$    and assigns odd-length names to all $\mathfrak{a}$ with $\gneg p(\mathfrak{a})$. Classical logic exclusively deals with objects of a universe without a need for also considering names for them, as it is not concerned with decidability or computability. CoL, on the other hand, with its computational semantics, inherently calls for being more careful about differentiating between objects and their names, and hence for explicitly considering universes in the form $(U, ^U)$ rather than just $U$ as classical logic does.

By a {\bf valuation}\label{ivaluation} on a universe $U$, or a $U$-valuation,  we mean 
a mapping $e$ that sends each variable $x$ to an element $e(x)$ of $U$. For technical convenience, we extend every such mapping to all constants as well, by stipulating that, for any constant $c$, $e(c)=c^U$. When a universe $U$ is fixed,  irrelevant or clear from the context, we may omit references to it and simply say ``valuation''. In these terms, a classical predicate $p$ can be understood as 
a function that sends each valuation $e$ to a proposition, i.e., to a constant predicate.   Similarly, what we call a game sends valuations to constant games: 

\begin{definition}\label{ngame}
%\marginpar{ngame}
Let $U$ be a universe. A {\bf game on $U$} is a function $A$ from $U$-valuations   to constant games. We write $e[A]$\label{iea} (rather than $A(e)$) to denote the constant game returned by $A$ on valuation $e$. Such a constant game $e[A]$ is said to be an {\bf instance}\label{iinstance} of $A$. 
For readability, we usually write $\legal{A}{e}$\label{ilre} and $\win{A}{e}$ instead of $\legal{e[A]}{}$ and $\win{e[A]}{}$.
\end{definition}

Just as this is the case with propositions versus predicates, constant games in the sense of Definition \ref{game} will
be thought of as special, constant cases of games in the sense of Definition \ref{ngame}. In particular, each constant game $A'$ is the game $A$ such that, for every valuation $e$,
$e[A]= A' $. From now on we will no longer distinguish between such $A$ and $A' $, so that, if $A$ is a constant game,
it is its own instance, with $A= e[A]$ for every $e$.

Where $n$ is a natural number, we say that a game $A$ is {\bf $n$-ary}\label{igarity} iff there is are $n$ variables such that, for any two valuations $e_1$ and $e_2$ that agree on all those variables, we have $e_1[A]= e_2[A]$. Generally, a game that is $n$-ary for some $n$, is said to be {\bf finitary}.\label{ifinitary} Our paper is going to exclusively deal with finitary games and, for this reason, we agree that, from now on, when we say ``game'', we always mean ``finitary game''.  

For a variable $x$ and valuations  $e_1,e_2$, we write $e_1\equiv_x e_2$ to mean that the two valuations have the same universe  and agree on all variables  other than $x$.

We say that a game $A$ {\bf depends} on a variable $x$ iff there are two valuations  $e_1,e_2$ with $e_1\equiv_x e_2$  such that $e_1[A]\not= e_2[A]$. An $n$-ary game thus depends on at most $n$ variables. And constant games are nothing but $0$-ary games, i.e., games that do not depend on any variables. 

We say that a (not necessarily constant) game $A$ is {\bf elementary}\label{ielem2} iff so are all of its instances $e[A]$. 

%And we say that $A$ is {\bf finite-depth}\label{fdpth} iff there is a (smallest) integer $d$, called the {\bf depth} of $A$, such that no legal run of any instance of $A$ has more than  $d$ moves.

Just as constant games are generalized propositions, games can be treated as generalized predicates. Namely, we will view each predicate $p$ of whatever arity as  the same-arity elementary game such that, for every valuation $e$,
$\win{p}{e}\emptyrun=\pp$ iff $p$ is true at $e$.  
And vice versa: every elementary game $p$ will be viewed as the same-arity predicate which is true at a given valuation $e$ iff  $\win{p}{e}\emptyrun=\pp$.   
Thus, for us, ``predicate'' and ``elementary game'' are synonyms. Accordingly,  any standard terminological or notational conventions familiar from the literature for predicates also apply to them to them viewed as elementary games.

Just as the Boolean operations straightforwardly extend from propositions to all predicates, our operations 
$\gneg,\mlc,\mld,\mli,\adc,\add,\st,\intimpl$ extend from constant games to all games. This is done by simply stipulating that $e[\ldots]$ commutes with all of those operations: $\gneg A$ is 
the game such that, for every valuation $e$, $e[\gneg A]=\gneg e[A]$; $A\adc B$ is the game such that,
for every $e$, $e[A\adc B]= e[A]\adc e[B]$; etc. So does the operation of prefixation: provided that $\Phi$ is a legal position of every instance of $A$,  $\seq{\Phi}A$ is  understood as the unique game such that, for every $e$, $e[\seq{\Phi}A]= \seq{\Phi}e[A]$.

\begin{definition}\label{sov}
%\marginpar{sov}
Let $A$ be a game on a universe $U$, $x_1,\ldots,x_n$ be pairwise distinct variables, and $\mathfrak{t}_1,\ldots,\mathfrak{t}_n$ be  constants and/or variables. 
The result of {\bf substituting $x_1,\ldots,x_n$ by $\mathfrak{t}_1,\ldots,\mathfrak{t}_n$ in $A$}, denoted $A(x_1/\mathfrak{t}_1,\ldots,x_n/\mathfrak{t}_n)$, is defined by stipulating that, for every valuation $e$, $e[A(x_1/\mathfrak{t}_1,\ldots,x_n/\mathfrak{t}_n)]= e'[A]$, where $e'$ is the valuation that sends each $x_i$ to $e(\mathfrak{t}_{i})$ and agrees with $e$ on all other variables. 
\end{definition}

Following the standard readability-improving practice established in the literature for predicates, we will often fix pairwise distinct  variables $x_1,\ldots,x_n$ for a game $A$ and write $A$ as $A(x_1,\ldots,x_n)$. 
Representing $A$ in this form  sets a context in which we can write $A(\mathfrak{t}_1,\ldots,\mathfrak{t}_n)$ to mean the same as the more clumsy expression $A(x_1/\mathfrak{t}_1,\ldots,x_n/\mathfrak{t}_n)$. 

\begin{definition}\label{bq}
%\marginpar{bq} 
Let  $A(x)$ be a game on a given universe. On the same universe,   $\ada xA(x)$ ({\bf choice universal quantification}) and  $\ade xA(x)$ ({\bf choice existential quantification}) are defined as the following two games, respectively:    
\[\begin{array}{l}
A(0)\adc A(1)\adc A(10)\adc A(11)\adc A(100)\adc \ldots;\\  
A(0)\add A(1)\add A(10)\add A(11)\add A(100)\add \ldots .   
\end{array}\]
\end{definition}

Thus, every initial legal move of $\ada xA(x)$ or $\ade xA(x)$ is a constant $c\in\{0,1,10,11,100,\ldots\}$,  which in our informal language we may refer to as ``the constant chosen (by the corresponding player) for $x$''. 

 We say that a game $A$ is  {\bf unistructural}\label{iunistructural} iff, for any two valuations $e_1$ and $e_2$,   we have $\legal{A}{e_1}= \legal{A}{e_2}$. Of course, all constant or elementary games are unistructural. It can also be easily seen that all our game operations preserve the unistructural property of games. For the purposes of the present paper, considering only  unistructural games would be sufficient. 
%Correspondingly, for simplicity and safety, we shall  limit our subsequent definitions to such games, even if those definitions are applicable to a wider classes of games or all games.

We define the remaining operations $\cla$ and $\cle$ only for unistructural games:

\begin{definition}\label{op5} Below  $A(x)$ is an arbitrary unistructural game on a universe $U$. On the same universe:\vspace{9pt}
%\marginpar{op5}

\noindent 1. The game $\cla x A(x)$ ({\bf blind universal quantification}) is defined by stipulating that, for every $U$-valuation $e$, player $\xx$ and move $\alpha$, we have: 
\begin{quote}\begin{description}
\item[(i)] $\seq{\xx\alpha}\in\legal{\clai x A(x)}{e}$ iff $\seq{\xx\alpha}\in\legal{A(x)}{e}$. Such an initial legal labmove $\xx\alpha$ brings the game $e[\cla x A(x)]$ down to 
$e[\cla x\seq{\xx\alpha}A(x)]$.
\item[(ii)] Whenever $\Gamma$ is a legal run of $e[\cla x A(x)]$,  $\win{\clai x A(x)}{e}\seq{\Gamma}= \pp$ iff, for every valuation $g$ with $g\equiv_x e$, $\win{A(x)}{g}\seq{\Gamma}= \pp$. \vspace{5pt}
\end{description}\end{quote}
\noindent 2. The game $\cle x A(x)$ ({\bf blind existential quantification})  is defined in exactly the same way, only with $\pp$ and $\oo$ interchanged.  
\end{definition}

\begin{example}\label{may14}
%\marginpar{may14}
Consider the ideal universe $U=\{0,1,10,11,100,\ldots\}$. Let $G$ be the following game on $U$, with the predicates {\em Even} and {\em Odd} having their expected meanings: 
\[\cla y\Bigl(\mbox{\em Even}(y)\add \mbox{\em Odd}(y) \mli \ada   x\bigl(\mbox{\em Even}(x\plus y)\add \mbox{\em Odd}(x\plus y)\bigr)\Bigr).\]
 Then the sequence 
$\seq{\oo 1.11,\ \oo 0.0,\ \pp 1.1}$ 
is a legal run of $G$, the effects of the moves of which are shown below:
\[\begin{array}{ll}
G:  & \cla y\Bigl(\mbox{\em Even}(y)\add \mbox{\em Odd}(y) \mli \ada   x\bigl(\mbox{\em Even}(x\plus y)\add \mbox{\em Odd}(x\plus y)\bigr)\Bigr)\\
\seq{\oo 1.11}G:  & \cla y\bigl(\mbox{\em Even}(y)\add \mbox{\em Odd}(y) \mli \mbox{\em Even}(11\plus y)\add \mbox{\em Odd}(11\plus y)\bigr)\\
\seq{\oo 1.11, \oo 0.0}G: &  \cla y\bigl(\mbox{\em Even}(y) \mli \mbox{\em Even}(11\plus y)\add \mbox{\em Odd}(11\plus y)\bigr)\\
\seq{\oo 1.11, \oo 0.0,\pp 1.1}G: & \cla y\bigl(\mbox{\em Even}(y) \mli \mbox{\em Odd}(11\plus y)\bigr)
\end{array}\]
The play hits (ends as) the true proposition $\cla y\bigl(\mbox{\em Even}(y) \mli \mbox{\em Odd}(11\plus y)\bigr)$ and hence is won by $\pp$. 

\end{example}

\begin{example}\label{newexample} 
The sequence 
\[\seq{\pp 0.1.\col{},\ \oo 1.10,\ \pp 0.1.0.10,\ \pp 0.1.0.10,\ \oo 0.1.0.100,\ \pp 0.1.1.100,\ \pp 0.1.1.10,\ \oo 0.1.1.1000,\ \pp 1.1000}\] 
is a legal run of the game 
\begin{equation}\label{mar7}
\cla x \bigl(x^3\equals (x\mult x)\mult x\bigr),\ \ada x\ada y \ade z (z\equals x\mult y) \ \intimpl \ \ada x\ade y(y\equals x^3).
\end{equation}
Below we see how the game evolves according to the scenario of this run: 

\[\begin{array}{ll}
& \cla x \bigl(x^3\equals (x\mult x)\mult x\bigr),\ \ada x\ada y \ade z (z\equals x\mult y) \ \intimpl \ \ada x\ade y(y\equals x^3)\\
\mbox{$\pp 0.1.\col{}$ yields} & \cla x \bigl(x^3\equals (x\mult x)\mult x\bigr),\ \ada x\ada y \ade z (z\equals x\mult y)\circ\ada x\ada y \ade z (z\equals x\mult y)  \ \intimpl \ \ada x\ade y(y\equals x^3)\\
 \mbox{$\oo 1.10$ yields} & \cla x \bigl(x^3\equals (x\mult x)\mult x\bigr),\ \ada x\ada y \ade z (z\equals x\mult y)\circ\ada x\ada y \ade z (z\equals x\mult y)  \ \intimpl \ \ade y(y\equals 10^3)\\
\mbox{$\pp 0.1.0.10$ yields} & \cla x \bigl(x^3\equals (x\mult x)\mult x\bigr),\ \ada y \ade z (z\equals 10\mult y)\circ\ada x\ada y \ade z (z\equals x\mult y)  \ \intimpl \ \ade y(y\equals 10^3)\\
 \mbox{$\pp 0.1.0.10$ yields} & \cla x \bigl(x^3\equals (x\mult x)\mult x\bigr),\ \ade z (z\equals 10\mult 10)\circ\ada x\ada y \ade z (z\equals x\mult y)  \ \intimpl \ \ade y(y\equals 10^3)\\
\mbox{$\oo 0.1.0.100$ yields} & \cla x \bigl(x^3\equals (x\mult x)\mult x\bigr),\ (100\equals 10\mult 10)\circ\ada x\ada y \ade z (z\equals x\mult y)  \ \intimpl \ \ade y(y\equals 10^3)\\
\mbox{$\pp 0.1.1.100$ yields} & \cla x \bigl(x^3\equals (x\mult x)\mult x\bigr),\ (100\equals 10\mult 10)\circ\ada y \ade z (z\equals 100\mult y)  \ \intimpl \ \ade y(y\equals 10^3)\\
\mbox{$\pp 0.1.1.10$ yields} & \cla x \bigl(x^3\equals (x\mult x)\mult x\bigr),\ (100\equals 10\mult 10)\circ \ade z (z\equals 100\mult 10)  \ \intimpl \ \ade y(y\equals 10^3)\\
\mbox{$\oo 0.1.1.1000$ yields} & \cla x \bigl(x^3\equals (x\mult x)\mult x\bigr),\ (100\equals 10\mult 10)\circ (1000\equals 100\mult 10)  \ \intimpl \ \ade y(y\equals 10^3)\\
 \mbox{$\pp 1.1000$ yields} &\cla x \bigl(x^3\equals (x\mult x)\mult x\bigr),\ (100\equals 10\mult 10)\circ (1000\equals 100\mult 10)  \ \intimpl \ 1000\equals 10^3

\end{array}\]
The play hits a true proposition and hence is won by $\pp$. Note that here, unlike the case in the previous example, $\top$ is the winner no matter what universe we consider and what the meanings of the expressions $x\times y$ and $x^3$ are. In fact, $\pp$ has a ``purely logical'' winning strategy in this game, in the sense that the strategy is successful  regardless of whether things have their standard arithmetic meanings or some other meanings. This follows from the promised soundness of {\bf CL12} and the fact --- illustrated later in Example \ref{ecube} --- that (\ref{mar7}) is provable in {\bf CL12}.

\end{example}

\section{Interactive machines revisited}\label{icp}
%\marginpar{icp}

In traditional game-semantical approaches, including Blass's \cite{Bla72,Bla92} approach which is the closest precursor of ours, player's strategies are understood as {\em functions} --- typically as functions from interaction histories (positions) to moves, or sometimes (\cite{Abr94}) as functions that only look at the latest move of the history. This {\em strategies-as-functions} approach, however, is generally inapplicable in the context of CoL, 
whose relaxed semantics, in striving to get rid of  ``bureaucratic pollutants'' and only deal with the remaining true essence of games,  does not impose any regulations on which player can or should move in a given situation. Here, in many cases, either player may have (legal) moves, and then it is unclear whether the next move should be the one prescribed by $\pp$'s strategy function or the one prescribed by the strategy function of $\oo$. For a game semantics whose ambition is to provide a comprehensive, natural and direct tool for modeling interaction, the strategies-as-functions approach would be less than adequate, even if technically possible. This is so for the simple reason that  the strategies that real computers follow are not functions. If the strategy of your personal computer was a function from the history of interaction with you, then its performance would keep noticeably worsening due to the need to read the continuously lengthening --- and, in fact, practically infinite --- interaction history every time before responding. Fully ignoring that history and looking only at your latest keystroke in the spirit of \cite{Abr94} is also certainly not what your computer does, either. The advantages of our approach thus become especially appreciable when one tries to bring complexity theory into interactive computation: hardly (m)any really meaningful and interesting complexity-theoretic concepts can be defined for games (particularly, games that may last long) with the strategies-as-functions approach.  

In CoL, ($\pp$'s effective) strategies are defined in terms of interactive machines, where computation is one continuous process interspersed with --- and influenced by --- multiple ``input'' (environment's moves) and ``output'' (machine's moves) events. Of several, seemingly rather different yet equivalent,  machine models of interactive computation studied in CoL, this paper only employs the most basic, {\bf HPM}\label{ihpm} (``Hard-Play Machine'') model.

Remember that an HPM   is   a Turing machine with the additional capability of making moves. The adversary can also move at any time, with such moves being the only nondeterministic events from the machine's perspective. Along with the ordinary read/write {\bf work tape}, the machine has an additional, read-only   tape called  the  {\bf run tape}.\footnote{The EPMs from the earlier literature on CoL also had a third tape called the {\em valuation tape}. The latter, however, becomes redundant in our present treatment due to the fact that we are exclusively interested in constant games or finitary games identified with their (constant) $\adc$-closures.} The latter, serving as a dynamic input, at any time  spells the ``current position'' of the play. Its role is to make the run fully visible to the machine. 
In these terms,  an  algorithmic solution ($\pp$'s winning strategy) for a given constant game $A$ is understood as an HPM $\cal M$ such that,  no matter how the environment acts during its interaction with $\cal M$ (what moves it makes and when),  the run incrementally spelled on the run tape is a $\pp$-won run of $A$.   
As for $\oo$'s strategies, there is no need to define them: all possible behaviors by $\oo$ are accounted for by the different possible nondeterministic updates  of the run tape of an HPM. 

In the above outline, we described HPMs in a relaxed fashion, without being specific about details  such as, say, how, exactly, moves are made by the machine, how many moves either player can make at once, what happens if both players attempt to move ``simultaneously'', etc. All reasonable design choices yield the same class of winnable games as long as we consider the natural subclass of games called {\bf static}.\label{istatic} Such games are obtained by imposing a certain simple formal condition on games (see the footnote on page \pageref{foot}, or Section 5 of \cite{Japfin}), which is not really necessary to reproduce here as nothing in this paper relies on it. We shall only point out that, intuitively, static games are interactive tasks where the relative speeds of the players are irrelevant, as it never hurts a player to postpone making moves. In other words, static games are games that are contests of intellect rather than contests of speed. And one of the theses that CoL philosophically relies on is that static games present an adequate formal counterpart of our intuitive concept of ``pure'', speed-independent interactive computational problems. Correspondingly, CoL restricts its attention (more specifically, possible interpretations of the atoms of its formal language) to static games. All elementary games turn out to be trivially static, and the class of static games turns out to be closed under all game operations studied in CoL. More specifically, all games expressible in the language of the later-defined logic $\cltw$  are static (as well as finitary and unistructural). Such games are not necessarily constant but, due to being finitary, can and will be thought of to be constant by identifying them with their $\ada$-closures. Correspondingly, in this paper, we use the 
 term ``{\bf computational problem}", or simply ``{\bf problem}", as a synonym of ``constant static game''.

While design choices are unimportant and ``negotiable'', we still want to agree on some technical details for clarity.   
Just like an ordinary Turing machine, an HPM has a finite set of {\bf states},\label{istate} one of which has the special status of being the {\bf start state}. There are no accept, reject, or halt states, but there are specially designated states called {\bf move states}.\label{imovestate}  
Either tape of the machine has a beginning but no end and is divided into infinitely many {\bf cells},\label{icell} arranged in the left-to-right order. At any time, each cell  contains one symbol from a certain fixed finite set of {\bf tape symbols}.\label{itapesymbol} The {\bf blank} symbol, as well as $\pp$ and $\oo$, are among the tape symbols. 
We also assume that these three symbols  are not among the symbols that any (legal or illegal) move can ever contain.  
Either tape has its own {\bf scanning head},\label{ihead} at any given time looking (located) at one of the cells of the tape.  A transition from one {\bf computation step}  (``{\bf clock cycle}'', {\bf ``time}'')   to another happens according to the fixed {\bf transition function}\label{itf} of the machine. The latter, depending on the current state, and the symbols seen by the two heads on the corresponding tapes, deterministically prescribes the next state, the tape symbol by which the old symbol should be overwritten in the current cell   (the cell currently scanned by the  head) of the work tape, and, for each head, the direction --- one cell left or one cell right --- in which the head should move. A constraint here is that the blank symbol, $\pp$ or $\oo$ can never be written by the machine on the work tape. An attempt to move left when the head of a  given   tape is looking at the first (leftmost) cell  results in staying put. So does an attempt 
to move right when the head is looking at the blank symbol. 

When the machine starts working, it is in its start state, both scanning heads are looking at the leftmost cells of the corresponding tapes,  and (all cells of) both   tapes are blank (i.e., contain the blank symbol). Whenever the machine enters a move state, the string $\alpha$ spelled by (the contents of) its work tape cells, starting from the first cell and ending with the cell immediately left to the work-tape scanning head,   will be automatically appended --- at the beginning of the next clock cycle --- to the contents of the run tape in the $\pp$-prefixed form  $\pp\alpha$. And, on every transition, whether the machine is in a move state or not, any finite sequence $\oo\beta_1,\ldots,\oo\beta_m$ of $\oo$-labeled moves may be nondeterministically appended to the content of the run tape. If the above two events happen on the same clock cycle, then the moves will be appended to the contents of the run tape in the following order: $\pp\alpha\oo\beta_1\ldots\oo\beta_m$.  

With each labmove  that emerges on the run tape we associate its {\bf timestamp},  which is the number of the clock cycle immediately preceding the cycle on which the move first emerged on the run tape. Intuitively, the timestamp indicates on which cycle the move was {\em made} rather than {\em appeared} on the run tape: a move made during cycle $\#i$ appears on the run tape on cycle $\#i\plus 1$ rather than $\#i$. Also, we agree that the count of clock cycles starts from $0$, meaning that the very first clock cycle is cycle $\#0$ rather than $\#1$. 

A {\bf configuration}\label{iconfiguration} is a full description of (the ``current'') contents of the work and run tapes, the locations of the two scanning heads, and the state of the machine. 
A {\bf computation branch}\label{icb} is an infinite sequence $C_0,C_1,C_2,\ldots$ of configurations, where $C_0$ is the initial configuration (as explained earlier), and every $C_{i\plus 1}$ is a configuration that could have legally followed (again,  in the sense explained earlier) $C_i$.  For a computation branch $B$, the {\bf run spelled by $B$}\label{irsb} is the run $\Gamma$ incrementally spelled on the run tape in the corresponding scenario of interaction. We say that such a $\Gamma$ is {\bf a run generated by} \label{irgb} the machine.  

We say that a given HPM $\cal M$ {\bf wins} ({\bf computes}, {\bf solves}) a given constant game $A$  --- and write ${\cal M}\models A$\label{imodels} --- iff every run $\Gamma$ generated by $\cal M$  is a $\pp$-won run of $A$. We say that $A$ is {\bf computable}\label{icomputable} iff there is an HPM $\cal M$ with ${\cal M}\models A$; such an HPM is said to be an (algorithmic) {\bf solution},\label{isol} or {\bf winning strategy}, for $A$.  

\section{Towards interactive complexity}\label{s7}
%\marginpar{s7}

At present, the theory of interactive computation is far from being well developed, and even more so is the corresponding complexity theory. The studies of interactive computation in the context of complexity, while having going on since long ago, have been relatively scattered and ad hoc:  more often than not,  interaction 
has  been used for better understanding certain complexity issues for traditional, non-interactive problems  rather than being treated as an object of systematic studies 
in its own rights (examples would be alternating computation \cite{Chandra}, or   interactive proof systems and Arthur-Merlin games \cite{Goldwasser,Babai}). 
 As if complexity theory was not ``complex'' enough already, taking it to the interactive level would most certainly generate a by an order of magnitude greater diversity of species from the complexity zoo. 
The present work is the first modest attempt to bring complexity issues into CoL.  Here we introduce one, perhaps the simplest, way of measuring the complexity of (our non-functional) strategies out of the huge and interesting potential variety of complexity measures meaningful and useful in the interactive context.

The {\bf size} of a move $\alpha$ means the length of $\alpha$ as a string.  In the context of a given computation branch of a given HPM $\cal M$, by the 
{\bf background} of a clock cycle $c$ we mean the greatest of the sizes of Environment's moves made by (before) time $c$, or $0$ if there are no such moves. 
If $\cal M$ makes a move on cycle $c$, then the background of that move\footnote{As easily understood, here and in similar contexts,  ``move'' means a move not as a {\em string}, but as an {\em event}, namely, the event of $\cal M$ making a move at time $c$.} means the background of $c$. 
Next, whenever $\cal M$ makes a move on cycle $c$, by  the {\bf timecost} of that move  we mean  $c\minus d$, where $d$ is the greatest cycle with $d< c$ on which a move was made by either player, or is $0$ if there is no such cycle.    

Throughout this paper, an $n$-ary ($n\geq 0$)  {\bf arithmetical function} means a total function from $n$-tuples of natural numbers to natural numbers. ``Unary'' is a synonym of ``$1$-ary''. 

\begin{definition}\label{deftcs}
%\marginpar{deftcs}
Let  $h$ be an unary arithmetical  function, and $\cal M$ an HPM. 

1. We say that {\bf $\cal M$ runs in time $h$}, or that $\cal M$ is an {\bf $h$ time machine}, or that $h$ is a {\bf bound} for the time complexity of $\cal M$, iff, in every play (computation branch), for  any clock cycle $c$ on which $\cal M$ makes a move, neither the timecost nor the size of that move  exceeds   $h(\ell)$,  where $\ell$ is the background of $c$.  

2. We say that {\bf $\cal M$ runs in space $h$}, or that $\cal M$ is an {\bf $h$ space machine}, or that $h$ is a {\bf bound} for the space complexity of $\cal M$, iff,  in every play (computation branch), for  any clock cycle $c$,    the number of cells ever visited by the work-tape head of $\cal M$ by time $c$ does not exceed $h(\ell)$, where $\ell$ is the background of $c$.       

\end{definition}

Our time complexity concept can be seen to be in the spirit of what is usually called {\em response time}. The latter generally does not and should not depend on the length of the preceding interaction history. On the other hand, it is not and should not  merely be a function of the adversary's last move, either. A   similar characterization applies to our concept of space complexity. Both complexity measures are equally  meaningful whether it be in the context of ``short-lasting'' games (such as the ones represented by the formulas  of the later-defined logic $\cltw$) or  the context of games that may have ``very long'' and even infinitely long legal runs.

Let $A$ be a constant game, $h$ an unary arithmetical function, and $\cal M$ an HPM. We say that {\bf $\cal M$ wins} ({\bf computes}, {\bf solves}) {\bf $A$ in time $h$}, or that {\bf $\cal M$ is an $h$ time solution for $A$},  iff $\cal M$ is an $h$ time machine with ${\cal M}\models A$. 
We say that $A$ is {\bf computable} ({\bf winnable}, {\bf solvable}) {\bf in time $h$} iff it has an $h$ time solution. Similarly for ``{\bf space}'' instead of ``time''. 

When we say {\bf polynomial time}, it is to be understood as ``time $h$ for some polynomial function $h$''. Similarly for {\bf polynomial space}.

\section{The language of logic $\cltw$ and its semantics}\label{ss6}
%\marginpar{ss6}

Logic $\cltw$ will be axiomatically constructed in Section \ref{ss8}. The present section is merely devoted to its {\em language}. The building blocks of the formulas of the latter are:

\begin{itemize} 
\item {\bf Nonlogical predicate letters},\label{ipl} for which we use $p,q$  as metavariables. With each predicate letter is associated a fixed nonnegative integer called its {\bf arity}.\label{iar2} We assume that, for any $n$, there are infinitely many $n$-ary predicate letters.   
\item {\bf Function letters},\label{ifl} for which we use $f,g$ as metavariables. Again, each function letter comes with a fixed {\bf arity},\label{iar3} and  we assume that, for any $n$, there are infinitely many $n$-ary function letters.  
\item The binary {\bf logical predicate letter} $\equals $.
\item Infinitely many {\bf variables} and {\bf constants}.  These are the same as the ones fixed in Section \ref{nncg}.
\end{itemize}

{\bf Terms},\label{ipterm} for which we use $\tau,\psi,\xi$  as metavariables,  are built from variables, constants and function letters in the standard way.  An {\bf atomic formula} is $p(\tau_1,\ldots,\tau_n)$, where $p$ is an $n$-ary predicate letter and the $\tau_i$ are terms. 
When $p$ is $0$-ary, we write $p$ instead of $p()$. Also, we write $\tau_1\equals \tau_2$ instead of $\equals (\tau_1,\tau_2)$, and $\tau_1\notequals \tau_2$ instead of $\gneg (\tau_1\equals \tau_2)$. 
{\bf Formulas} are built from atomic formulas, propositional connectives $\twg,\tlg$ ($0$-ary), $\gneg$ ($1$-ary), $\mlc,\mld,\adc,\add$ ($2$-ary), variables and quantifiers $\cla,\cle,\ada,\ade$  in the standard way, with the exception that, officially, $\gneg$ is only allowed to be applied to atomic formulas. The definitions of {\em free} and {\em bound} occurrences of variables are standard  (with $\ada,\ade$ acting as quantifiers along with $\cla,\cle$). A formula with no free occurrences of variables is said to be {\bf closed}.

Note that, terminologically, $\twg$ and $\tlg$ do not count as atoms. For us, atoms are formulas containing no logical operators. The formulas $\twg$ and $\tlg$ do not qualify because they {\em are} ($0$-ary) logical operators themselves.

$\gneg E$, where $E$ is not atomic, will be understood as a standard abbreviation: 
$\gneg\twg=\tlg$, $\gneg\gneg E= E$, $\gneg(A\mlc B)= \gneg A\mld \gneg B$, $\gneg \ada xE= \ade x\gneg E$, etc. And $E\mli F$ will be understood as an abbreviation of $\gneg E\mld F$. 
%Also, if we write \[E_1\mli E_2\mli E_3\mli\ldots\mli E_n,\] this is to be understood as an abbreviation of $E_1\mli(E_2\mli(E_3\mli(\ldots (E_{n-1}\mli E_n)\ldots)))$.

Parentheses will often be omitted --- as we just did --- if there is no danger of ambiguity. When omitting parentheses, we assume that $\gneg$ and the quantifiers have the highest precedence, and $\mli$ has the lowest precedence. %So, for instance, $\gneg \ada x E \mli F\mld G$ means $(\gneg(\ada x(E)))\mli((F)\mld (G))$. 
%We further assume that $\adc,\add$ take precedence over $\mlc,\mld$, so that, for instance, if we write $E\mlc F\adc G$, it should be read as  $E\mlc (F\adc G)$. 
An expression $E_1\mlc\ldots\mlc E_n$, where $n\geq 2$, is to be understood as $E_1\mlc(E_2\mlc (\ldots\mlc(E_{n-1}\mlc E_n)\ldots)$. Sometimes we can write this  expression for an unspecified $n\geq 0$ (rather than $n\geq 2$). Such a formula, in the case of $n= 1$, should be understood as simply $E_1$. Similarly for $\mld,\adc,\add$.  As for the case of $n=0$, $\mlc$ and $\adc$ should be understood as $\twg$ while $\mld$ and $\add$ as $\tlg$.  
%Also, where $\cal F$ is a set of formulas, we may write \(\mlc {\cal F}\) for the $\mlc$-conjunction of the elements of $\cal F$. Similarly for $\mld,\adc,\add$.    

%The expressions $\vec{x},\vec{y},\ldots$ will usually stand for tuples of variables. Similarly for $\vec{\tau},\vec{\theta},\ldots$ (for tuples of terms) or $\vec{a},\vec{b},\ldots$ (for tuples of constants). 
  
 %We will try to use $x,y,z$ for bound  variables only, while use $s,r,t,u,v,w$ for free variables only. There may be some occasional violations of this commitment though.

Sometimes a formula $F$ will be represented as $F(s_1,\ldots,s_n)$, where the $s_i$ are variables. 
When doing so, we do not necessarily mean that each  $s_i$ has a free occurrence in $F$, or that every variable occurring free in $F$ is among $s_1,\ldots,s_n$. However, it {\em will}  always be assumed (usually only implicitly) that the $s_i$ are pairwise distinct, and have no bound occurrences in $F$.  In the context set by the above representation, $F(\tau_1,\ldots,\tau_n)$ will mean the result of replacing, in $F$, each  occurrence of each $s_i$   by term $\tau_i$. When writing $F(\tau_1,\ldots,\tau_n)$, it will always be assumed (again, usually only implicitly) that the terms $\tau_1,\ldots,\tau_n$ contain no variables that have bound occurrences in $F$, so that there are no unpleasant collisions of variables when doing replacements.  

Similar --- well established in the literature --- notational conventions apply to terms.

A {\bf sequent} is an expression $E_1,\ldots,E_n\intimpl F$, where $E_1,\ldots,E_n$ ($n\geq 0$) and $F$ are formulas. Here $E_1,\ldots,E_n$ is said to be the {\bf antecedent} of the sequent, and $F$ said to be the {\bf succedent}. 

By a {\bf free} (resp. {\bf bound}) {\bf variable} of a sequent we shall mean a variable that has a free (resp. bound) occurrence in one of the formulas of the sequent. For safety and simplicity, throughout the rest of this paper we assume that the sets of all free and bound variables of any sequent that we ever consider --- unless strictly implied otherwise by the context ---  are disjoint.   This restriction, of course, does not yield any loss of expressive power as variables can always be renamed so as to satisfy this condition. 

An {\bf interpretation}\label{iint}\footnote{The concept of an interpretation in CoL is usually more general than the present one. Interpretations in our present sense are called  {\bf perfect}. But here we omit the word ``perfect'' as we do not consider any nonperfect interpretations, anyway.} is a pair $(U,^*)$, where $U=(U,^U)$ is a universe and $^*$ is a function that sends:
\begin{itemize}
\item  every $n$-ary function letter $f$ to a function \(f^*:\ U^n\rightarrow U\);
\item  every nonlogical $n$-ary  predicate letter $p$ to an $n$-ary predicate (elementary game) $p^*(s_1,\ldots,s_n)$ on $U$ which does not depend on any variables other than $s_1,\ldots,s_n$. 
 \end{itemize}

The above uniquely extends to  a  mapping that sends each term $\tau$ to a function $\tau^*$, and each formula or sequent $S$ to a game $S^*$, by stipulating that: 
\begin{enumerate}
\item $c^*=c^U$ (any constant $c$).
\item $s^* =  s$ (any variable $s$). 
\item Where $f$ is an $n$-ary function letter and $\tau_1,\ldots,\tau_n$ are terms, $\bigl(f(\tau_1,\ldots,\tau_n)\bigr)^* =  f^*(\tau_{1}^{*},\ldots,\tau_{n}^{*})$. 
\item Where   $\tau_1$ and $\tau_2$ are terms, $(\tau_1\equals \tau_2)^*$ is $\tau_{1}^{*}\equals \tau_{2}^{*}$. 
\item Where $p$ is an $n$-ary nonlogical  predicate letter  and $\tau_1,\ldots,\tau_n$ are terms, $\bigl(p(\tau_1,\ldots,\tau_n)\bigr)^* =  p^*(\tau_{1}^{*},\ldots,\tau_{n}^{*})$. 
\item $^{*}$ commutes with all logical operators, seeing them as the corresponding game operations: $\tlg^* =  \tlg$,  $(E_1\mlc\ldots\mlc E_n)^{*} =  E^{*}_{1}\mlc \ldots\mlc E^{*}_n$, $(\ada x E)^{*} =  \ada x(E^{*})$, etc.
\item Similarly, $^*$ sees the sequent symbol $\intimpl$ as the same-name game operation, that is, $(E_1,\ldots,E_n\intimpl F)^* = E_{1}^{*},\ldots,E_{n}^{*}\intimpl F^{*}$.  
\end{enumerate}

While an interpretation is a pair $(U,^*)$, terminologically and notationally we will usually identify it with its second component and write $^*$ instead of $(U,^*)$, keeping in mind that every such ``interpretation'' $^*$ comes with a fixed universe $U$, said to be the {\bf universe of $^*$}.  When $O$ is a function letter, a predicate letter, a constant or a formula, and $O^* =  W$, we say that $^*$ {\bf interprets} $O$ as $W$. We can also refer to such a $W$ as 
``{\bf $O$ under interpretation $^*$}''.

When a given formula is represented as $F(x_1,\ldots,x_n)$, we will typically write $F^*(x_1,\ldots,x_n)$ instead of 
$\bigl(F(x_1,\ldots,x_n)\bigr)^*$. A similar practice will be used for terms as well.

We agree that, for a sequent or formula $S$, an interpretation $^*$ and an HPM $\cal M$, whenever we say that $\cal M$ is a {\bf  solution} of $S^*$ or write ${\cal M}\models S^*$, we mean that   $\cal M$ is a  solution of the (constant) game $\ada x_1\ldots\ada x_n (S^*)$, where $x_1,\ldots,x_n$ are exactly the free variables of $S$, listed according to their lexicographic order. We call the above game the {\bf $\ada$-closure} of $S^*$, and denote it by 
$\ada S^*$. 

Note that, for any given sequent or formula $S$, the $\legal{}{}$ component of the game $\ada S^*$ does not depend on the interpretation $^*$. Hence we can safely 
say ``legal run of $\ada S$'' --- or even just ``legal run of $S$'' --- without indicating an interpretation applied to the sequent.

We say that an HPM $\cal M$ is a {\bf uniform solution}, or a {\bf logical solution},  of a sequent $X$ iff, for any interpretation $^*$, ${\cal M}\models X^*$. 

Intuitively, a logical  solution is (indeed)  a ``purely logical'' solution. ``Logical'' in the sense that it does not depend on the universe and the meanings of the nonlogical symbols (predicate and function letters) ---  does not depend on a (the) interpretation $^*$, that is. It is exactly these kinds of   solutions that we are interested in when seeing CoL as a logical basis for applied theories or knowledge base systems. As a universal-utility tool, CoL (or a CoL-based compiler) would have no knowledge of the meanings of those nonlogical symbols (the meanings that will be changing from application to application and from theory to theory), other than what is explicitly  given by the target formula and the axioms or the knowledge base   of the system. 

\section{Logic $\cltw$}\label{ss8}
%\marginpar{ss8}

The purpose of the deductive system $\cltw$ that we construct in this section is to axiomatize the set of   sequents with logical solutions. 
 Our formulation of the system relies on the terminology and notation explained below.

\begin{enumerate}
\item A {\bf surface occurrence} of a subformula is an occurrence that is 
not in the scope of any choice operators ($\adc,\add,\ada$ and/or $\ade$). 
\item A formula not containing choice operators --- i.e., a formula of the language of classical first order logic --- is said to be {\bf elementary}. 
\item A sequent is {\bf elementary} iff all of its formulas are so. 
\item The {\bf elementarization} \[\elz{F}\] of a formula $F$ is the result of replacing
in $F$ all $\add$- and $\ade$-subformulas by $\tlg$, and all $\adc$- and $\ada$-subformulas by $\twg$. Note that $\elz{F}$ is (indeed) an elementary formula.
\item The {\bf elementarization} $\elz{G_1,\ldots,G_n\intimpl F}$ of a sequent 
$G_1,\ldots,G_n\intimpl F$ is the elementary formula \[\elz{G_1}\mlc\ldots\mlc \elz{G_n}\mli \elz{F}.\] 
\item A sequent  is said to be {\bf stable} iff its elementarization is classically valid; otherwise it is {\bf unstable}.  By ``classical validity'', in view of G\"{o}del's completeness theorem,  we mean provability in classical first-order calculus with constants, function letters and $\equals$, where $\equals$ is treated as the logical {\em identity} predicate (so that, say, $x\equals x$, $x\equals y\mli (E(x)\mli E(y))$, etc. are provable).
\item We will be using the notation \[F[E]\] to mean a formula $F$ together with some (single) fixed  surface occurrence of a subformula $E$. Using this notation sets a context, in which $F[H]$ will mean the result of replacing in $F[E]$ the (fixed) occurrence of $E$ by $H$.  Note that here we are talking about some {\em occurrence} of $E$. Only that occurrence gets replaced when moving from $F[E]$ to $F[H]$, even if the formula also had some other occurrences of $E$.
\item By a {\bf rule} (of inference) in this section we mean a binary relation $\mathbb{Y}{\cal R} X$, where $\mathbb{Y}=\seq{Y_1,\ldots,Y_n}$ is a finite sequence of sequents and $X$ is a sequent. Instances of such a relation are schematically written as 
\[\frac{Y_1,\ldots,Y_n}{X},\]
where $Y_1,\ldots,Y_n$ are called the {\bf premises}, and $X$ is  called the {\bf conclusion}. Whenever $\mathbb{Y}{\cal R}X$ holds, we say that $X$ {\bf follows} from $\mathbb{Y}$ by $\cal R$.  
\item Expressions such as $\vec{G},\vec{K},\ldots$ will usually stand for finite sequences of formulas. The standard meaning of an expression such as $\vec{G},F,\vec{K}$  should also be clear. 
\end{enumerate}

\begin{center}
\begin{picture}(100,30)

\put(0,10){\bf THE RULES OF $\cltw$}

\end{picture}
\end{center}

$\cltw$ has the six rules listed below, with the following additional conditions/explanations: 
\begin{enumerate}
\item In $\add$-Choose and $\adc$-Choose, $i\in\{0,1\}$.
\item In $\ade$-Choose and $\ada$-Choose,  $\mathfrak{t}$ is either a constant or a variable with no bound occurrences in the premise, and $H(\mathfrak{t})$ is the result of replacing by $\mathfrak{t}$ all free occurrences of $x$ in $H(x)$ (rather than vice versa).
\end{enumerate}
\begin{center}
\begin{picture}(287,70)

\put(14,50){\bf $\add$-Choose}
\put(12,30){$\vec{G}\ \intimpl\  F[H_i]$}
\put(0,22){\line(1,0){78}}
\put(0,8){$\vec{G}\ \intimpl \ F[H_0\add H_1]$}

\put(232,50){\bf $\adc$-Choose}
\put(212,30){$\vec{G},\ E[H_i],\ \vec{K}\  \intimpl \ F$}
\put(200,22){\line(1,0){113}}
\put(200,8){$\vec{G},\ E[H_0\adc H_1], \ \vec{K}\ \intimpl\ F$}

\end{picture}
\end{center}

\begin{center}
\begin{picture}(287,70)

\put(231,50){\bf $\ada$-Choose}
\put(210,30){$\vec{G},\ E[H(\mathfrak{t})],\ \vec{K}\ \intimpl\ F$}
\put(200,22){\line(1,0){113}}
\put(200,8){$\vec{G},\ E[\ada xH(x)],\ \vec{K}\ \intimpl\ F$}

\put(16,50){\bf $\ade$-Choose}
\put(08,30){$\vec{G}\ \intimpl\  F[H(\mathfrak{t})]$}
\put(0,22){\line(1,0){78}}
\put(0,8){$\vec{G}\ \intimpl \ F[\ade x H(x)]$}

\end{picture}
\end{center}

\begin{center}
\begin{picture}(74,70)

\put(12,50){\bf Replicate}
\put(8,8){$\vec{G},E,\vec{K}\intimpl F$}
\put(0,22){\line(1,0){69}}
\put(0,30){$\vec{G},E,\vec{K},E\intimpl F$}
\end{picture}
\end{center}

\begin{center}
\begin{picture}(300,70)
\put(140,50){\bf Wait}
\put(0,30){$Y_1,\ldots,Y_n$}
\put(0,22){\line(1,0){45}}
\put(55,20){($n\geq 0$), where all of the following five conditions are satisfied:}
\put(20,8){$X$}
\end{picture}
\end{center}

\begin{enumerate}
\item {\bf $\adc$-Condition:}  Whenever $X$ has the form $\vec{G}\intimpl F[H_0\adc H_1]$, both of the sequents $\vec{G}\intimpl F[H_0]$ and 
$\vec{G}\intimpl F[H_1]$ are among $Y_1,\ldots,Y_n$.
\item {\bf $\add$-Condition:} Whenever $X$ has the form $\vec{G},E[H_0\add H_1],\vec{K}\intimpl F$, both of the sequents $\vec{G},E[H_0],\vec{K}\intimpl F$ and 
$\vec{G},E[H_1],\vec{K}\intimpl F$ are among $Y_1,\ldots,Y_n$.
\item {\bf $\ada$-Condition:} Whenever $X$ has the form $\vec{G}\intimpl F[\ada xH(x)]$, for some variable $y$ not occurring in $X$, the sequent  $\vec{G}\intimpl F[H(y)]$ is among  $Y_1,\ldots,Y_n$. Here and below, $H(y)$ is the result of replacing by $y$ all free occurrences of $x$ in $H(x)$ (rather than vice versa).
\item {\bf $\ade$-Condition:} Whenever $X$ has the form $\vec{G},E[\ade xH(x)],\vec{K}\intimpl F$, for some variable $y$ not occurring in $X$, the sequent  $\vec{G},E[H(y)],\vec{K}\intimpl F$ is among  $Y_1,\ldots,Y_n$.
\item {\bf Stability condition:} $X$ is stable.
\end{enumerate}

As will be seen in Section \ref{ssoundness}, each rule --- seen bottom-up --- encodes an action that a winning strategy should take in a corresponding situation, and the name of each rule is suggestive of that action. For instance, Wait (indeed) prescribes the strategy to wait till the adversary moves. This explains why we have called  ``Replicate'' the rule which otherwise is nothing but what is commonly known as Contraction.   

A {\bf $\cltw$-proof} of a sequent $X$ is a sequence $X_1,\ldots,X_n$ of sequents, with $X_n=X$, such that, each $X_i$ follows  by one of the rules of $\cltw$ from some (possibly empty in the case of Wait, and certainly empty in the case of $i=1$) set $\cal P$ of premises such that ${\cal P}\subseteq \{X_1,\ldots, X_{i-1}\}$.
When a $\cltw$-proof of $X$ exists, we say that $X$ is {\bf provable} in $\cltw$, and write $\cltw\vdash X$.

   A {\bf $\cltw$-proof} of a formula $F$ will be understood as a  $\cltw$-proof of the empty-antecedent sequent $\intimpl F$. Accordingly, $\cltw\vdash F$ means $\cltw\vdash\intimpl F$.

\begin{fact}\label{fce}
%\marginpar{fce}
$\cltw$ is a conservative extension of classical logic. That is, an elementary sequent {\em $E_1,\ldots,E_n\intimpl F$} is provable in $\cltw$ iff the formula $E_1\mlc\ldots\mlc E_n\mli F$ is valid in the classical sense.
\end{fact}

\begin{proof} Assume $E_1,\ldots,E_n,F$ are elementary formulas. If $E_1\mlc\ldots\mlc E_n\mli F$ is classically valid, then $E_1,\ldots,E_n\intimpl F$ follows   from the empty set of premises by Wait. And if $E_1\mlc\ldots\mlc E_n\mli F$ is not classically valid, then $E_1,\ldots,E_n\intimpl F$  cannot be the conclusion of any of the rules of $\cltw$ except Replicate. However, applying (bottom-up) Replicate does not take us any closer to finding a proof of the sequent, as the premise still remains an unstable elementary sequent.  
\end{proof}

$\cltw$ can also be seen to be a conservative extension of the earlier known logic {\bf CL3} studied in \cite{Japtcs}.\footnote{Essentially the same logic, called {\bf L}, was in fact known as early as in \cite{Jap02}.} The latter is nothing but the empty-antecedent fragment of $\cltw$ without function letters and identity. 

\begin{example}\label{ecube} In this example, $\mult$ is a binary function letter and $^3$ is a unary function letter. We write $x\mult y$ and $x^3$ instead of $\mult(x,y)$ and $^3(x)$, respectively. The following sequence of sequents is a $\cltw$-proof of the sequent (\ref{mar7}) from Example \ref{newexample}. It may be worth observing that the strategy used by $\top$ in that example, in a sense, ``follows'' our present proof step-by-step in the bottom-up direction. And this is no accident: as we are going to see in the course of proving the soundness of $\cltw$, every $\cltw$-proof rather directly encodes a winning strategy.\vspace{7pt}

\noindent 1. $\cla x \bigl(x^3\equals (x\mult x)\mult x\bigr),\    t\equals s\mult s, \  r\equals t\mult s \ \intimpl \  r\equals s^3$  \ \ {Wait: (no premises)} \vspace{3pt}

\noindent 2. $\cla x \bigl(x^3\equals (x\mult x)\mult x\bigr),\    t\equals s\mult s, \  r\equals t\mult s \ \intimpl \  \ade y(y\equals s^3)$  \ \   {$\ade$-Choose: 1}\vspace{3pt}

\noindent 3. $\cla x \bigl(x^3\equals (x\mult x)\mult x\bigr),\    t\equals s\mult s, \  \ade z (z\equals t\mult s) \ \intimpl \  \ade y(y\equals s^3)$  \ \ {Wait: 2} \vspace{3pt}

\noindent 4. $\cla x \bigl(x^3\equals (x\mult x)\mult x\bigr),\    t\equals s\mult s, \ \ada y \ade z (z\equals t\mult y) \ \intimpl \  \ade y(y\equals s^3)$  \ \ {$\ada$-Choose: 3}\vspace{3pt}

\noindent 5. $\cla x \bigl(x^3\equals (x\mult x)\mult x\bigr),\    t\equals s\mult s, \ \ada x\ada y \ade z (z\equals x\mult y) \ \intimpl \  \ade y(y\equals s^3)$  \ \ {$\ada$-Choose: 4}\vspace{3pt}

\noindent 6. $\cla x \bigl(x^3\equals (x\mult x)\mult x\bigr),\    \ade z (z\equals s\mult s), \ \ada x\ada y \ade z (z\equals x\mult y) \ \intimpl \  \ade y(y\equals s^3)$  \ \ {Wait: 5}\vspace{3pt}

\noindent 7. $\cla x \bigl(x^3\equals (x\mult x)\mult x\bigr),\  \ada y \ade z (z\equals s\mult y), \ \ada x\ada y \ade z (z\equals x\mult y) \ \intimpl \  \ade y(y\equals s^3)$  \ \ { $\ada$-Choose: 6}\vspace{3pt}

\noindent 8.  $\cla x \bigl(x^3\equals (x\mult x)\mult x\bigr),\ \ada x\ada y \ade z (z\equals x\mult y), \ \ada x\ada y \ade z (z\equals x\mult y) \ \intimpl \  \ade y(y\equals s^3)$ \ \ { $\ada$-Choose: 7}\vspace{3pt}

\noindent 9.  $\cla x \bigl(x^3\equals (x\mult x)\mult x\bigr),\ \ada x\ada y \ade z (z\equals x\mult y) \ \intimpl \ \ade y(y\equals s^3)$ \ \ {Replicate: 8}\vspace{3pt}

\noindent 10. $\cla x \bigl(x^3\equals (x\mult x)\mult x\bigr),\ \ada x\ada y \ade z (z\equals x\mult y) \ \intimpl \ \ada x\ade y(y\equals x^3)$ \ \ { Wait: 9}
\end{example}

\begin{example}\label{j28a}
%\marginpar{j28a}
The formula $\cla x\hspace{1pt}p(x)\mli\ada x\hspace{1pt}p(x)$ is provable in $\cltw$. It follows  from $\cla x\hspace{1pt}p(x)\mli p(y)$ by Wait. The latter, in turn, follows by Wait from the empty set of premises. 

On the other hand, the formula $\ada x\hspace{1pt}p(x)\mli\cla x\hspace{1pt}p(x)$, i.e. $\ade x\gneg p(x)\mld \cla x\hspace{1pt}p(x)$, in not provable. Indeed, its elementarization is $\tlg\mld \cla x\hspace{1pt}p(x)$, which is not classically valid.  Hence $\ade x\gneg p(x)\mld \cla x\hspace{1pt}p(x)$ cannot be derived by Wait. Replicate can also be dismissed for obvious reasons. This leaves us with $\ade$-Choose. But if $\ade x\gneg p(x)\mld \cla x\hspace{1pt}p(x)$ is derived  by $\ade$-Choose, then the premise should be $\gneg p(\mathfrak{t})\mld \cla x\hspace{1pt}p(x)$ for some variable or constant $\mathfrak{t}$. The latter, however, is a classically non-valid elementary formula and hence, by Fact \ref{fce}, is not provable. 
\end{example}

\begin{example}\label{j28b}
%\marginpar{j28b} 
The formula $\ada x\ade y\bigl(p(x)\mli p(y)\bigr)$ is provable in $\cltw$ as follows:\vspace{7pt}

\noindent 1. $\begin{array}{l}
p(s)\mli p(s)
\end{array}$  \ \ Wait:\vspace{3pt}

\noindent 2. $\begin{array}{l}
\ade y\bigl(p(s)\mli p(y)\bigr)
\end{array}$  \ \ $\ade$-Choose: 1\vspace{3pt}

\noindent 3. $\begin{array}{l}
\ada x\ade y\bigl(p(x)\mli p(y)\bigr)
\end{array}$  \ \ Wait: 2\vspace{7pt}

On the other hand, the formula $\ade y\ada x \bigl(p(x)\mli p(y)\bigr)$ can be seen to be unprovable, even though its classical counterpart $\cle y\cla x \bigl(p(x)\mli p(y)\bigr)$ is a classically valid elementary formula and hence provable in $\cltw$.  
\end{example}

\begin{example}\label{j28c}
%\marginpar{j28c} 
While the formula $\cla x\cle y \bigl(y\equals f(x)\bigr)$  is classically valid and hence provable in $\cltw$, its constructive counterpart 
$\ada x\ade y \bigl(y\equals f(x)\bigr)$ can be easily seen to  be unprovable. This is no surprise. In view of the expected soundness of $\cltw$,  provability  of $\ada x\ade y \bigl(y\equals f(x)\bigr)$ would imply that every function $f$ is computable, which, of course, is not the case.     
\end{example}

\begin{exercise}\label{feb1a}
%\marginpar{feb1a}
To see the resource-consciousness of $\cltw$, show that it does not prove $p\adc q\mli (p\adc q)\mlc (p\adc q)$, even though this formula has the form $F\mli F\mlc F$ of a classical tautology. Then show that, in contrast, $\cltw$ proves the {\em sequent} $p\adc q\intimpl (p\adc q)\mlc (p\adc q)$ because, unlike the antecedent of a $\mli$-combination, the antecedent of a $\intimpl$-combination is reusable (trough Replicate). 
\end{exercise}

\begin{exercise}\label{feb1ae}
%\marginpar{feb1ae}
Show that $\cltw\vdash \ade x\ada y\hspace{2pt} p(x,y)\intimpl \ade x\bigl(\ada y\hspace{2pt}p(x,y)\mlc \ada y\hspace{2pt}p(x,y)\bigr)$. Then observe that, on the other hand,  $\cltw$ does not prove any of the formulas 
\[\begin{array}{rcl}
\ade x\ada y\hspace{2pt} p(x,y) & \mli & \ade x\bigl(\ada y\hspace{2pt}p(x,y)\mlc \ada y\hspace{2pt}p(x,y)\bigr);\\
\ade x\ada y\hspace{2pt} p(x,y)\ \mlc \ \ade x\ada y\hspace{2pt} p(x,y) & \mli & \ade x\bigl(\ada y\hspace{2pt}p(x,y)\mlc \ada y\hspace{2pt}p(x,y)\bigr);\\
\ade x\ada y\hspace{2pt} p(x,y)\ \mlc\ \ade x\ada y\hspace{2pt} p(x,y)\ \mlc\ \ade x\ada y\hspace{2pt} p(x,y) & \mli & \ade x\bigl(\ada y\hspace{2pt}p(x,y)\mlc \ada y\hspace{2pt}p(x,y)\bigr);\\
 & \ldots & 
\end{array}\]
Intuitively, this contrast is due to the fact that, even though  both $\st A$ and $\pst A = A\mlc A\mlc A\mlc\ldots$ are resources allowing to reuse $A$ any number of times, the ``branching'' form of reusage offered by $\st A$ is substantially stronger than the ``parallel'' form of reusage offered by $\pst A$.  $\st \ade x\ada y\hspace{2pt} p(x,y)\mli \ade x\bigl(\ada y\hspace{2pt}p(x,y)\mlc \ada y\hspace{2pt}p(x,y)\bigr)$ is a valid principle of CoL while $\pst \ade x\ada y\hspace{2pt} p(x,y)\mli \ade x\bigl(\ada y\hspace{2pt}p(x,y)\mlc \ada y\hspace{2pt}p(x,y)\bigr)$ is not.  
\end{exercise}

\section{The soundness of $\cltw$}\label{ssoundness}
%\marginpar{ssoundness}

We say that a logical solution $\cal M$ of a sequent $X$ is {\bf well-behaved}  iff  the following  conditions are satisfied:
\begin{enumerate} 
\item There is an integer $d$ such that, in every play,  $\cal M$  makes at most $d$ replicative moves in the antecedent of $X$. 
\item Every non-replicative move that $\cal M$ makes in the antecedent of $X$ is focused. 
\item Every time when $\cal M$ chooses a constant $c$ for a variable $x$ in some $\ade x G$ or $\ada x G$ component of $X$, $c$ is either $0$, or a constant that occurs in $X$,  or a constant already chosen by Environment for some variable $y$ in some  $\ada y G$ or $\ade y G$ component of $\ada X$.
\end{enumerate}

The terms of the language of $\cltw$, identified with their parse trees, are tree-style structures, so let us call them {\bf tree-terms}. A more general and economical way to represent terms, however, is to allow merging some or all identical-content nodes in such trees, thus turning them into (directed, acyclic,  rooted, edge-ordered multi-) graphs.  Let us call these (unofficial) sorts of terms   {\bf graph-terms}. The idea of representing linguistic objects in the form of graphs rather than trees is central in the approach called {\em cirquent calculus} (\cite{Cirq,Japdeep}), and has already  proven its worth. We find that idea particularly useful in our present, complexity-sensitive context. Figure 1 illustrates two terms representing the same polynomial function $y^8$, with the term on the right being a tree-term and the term on the left being a graph-term. As this example suggests, graph-terms are generally exponentially smaller than the corresponding tree-terms, which explains our preference for the former as a standard way of writing (in our metalanguage) polynomial terms. Figure 1 also makes it unnecessary to formally define graph-terms, as their meaning must be perfectly clear after looking at this single example.   

\begin{center} \begin{picture}(353,130)

\put(16,112){\circle{12}}
\put(14,111){$y$}

\put(22,90){\line(1,1){10}}
\put(32,100){\vector(-1,1){10}}
\put(10,90){\line(-1,1){10}}
\put(0,100){\vector(1,1){10}}
\put(16,86){\circle{12}}
\put(12,83){$\mult$}

\put(22,64){\line(1,1){10}}
\put(32,74){\vector(-1,1){10}}
\put(10,64){\line(-1,1){10}}
\put(0,74){\vector(1,1){10}}
\put(16,60){\circle{12}}
\put(12,57){$\mult$}

\put(22,38){\line(1,1){10}}
\put(32,48){\vector(-1,1){10}}
\put(10,38){\line(-1,1){10}}
\put(0,48){\vector(1,1){10}}
\put(16,34){\circle{12}}
\put(12,31){$\mult$}

\put(104,112){\circle{12}}
\put(102,111){$y$}
\put(137,112){\circle{12}}
\put(135,111){$y$}
\put(116,90){\vector(-2,3){10}}
\put(125,90){\vector(2,3){10}}
\put(121,86){\circle{12}}
\put(117,83){$\mult$}

\put(174,112){\circle{12}}
\put(172,111){$y$}
\put(207,112){\circle{12}}
\put(205,111){$y$}
\put(186,90){\vector(-2,3){10}}
\put(195,90){\vector(2,3){10}}
\put(191,86){\circle{12}}
\put(187,83){$\mult$}

\put(151,64){\vector(-3,2){26}}
\put(161,64){\vector(3,2){26}}
\put(156,60){\circle{12}}
\put(152,57){$\mult$}

\put(244,112){\circle{12}}
\put(242,111){$y$}
\put(277,112){\circle{12}}
\put(275,111){$y$}
\put(256,90){\vector(-2,3){10}}
\put(265,90){\vector(2,3){10}}
\put(261,86){\circle{12}}
\put(257,83){$\mult$}

\put(314,112){\circle{12}}
\put(312,111){$y$}
\put(347,112){\circle{12}}
\put(345,111){$y$}
\put(326,90){\vector(-2,3){10}}
\put(335,90){\vector(2,3){10}}
\put(331,86){\circle{12}}
\put(327,83){$\mult$}

\put(291,64){\vector(-3,2){26}}
\put(301,64){\vector(3,2){26}}
\put(296,60){\circle{12}}
\put(292,57){$\mult$}

\put(221,38){\vector(-4,1){63}}
\put(231,38){\vector(4,1){63}}
\put(226,34){\circle{12}}
\put(222,31){$\mult$}

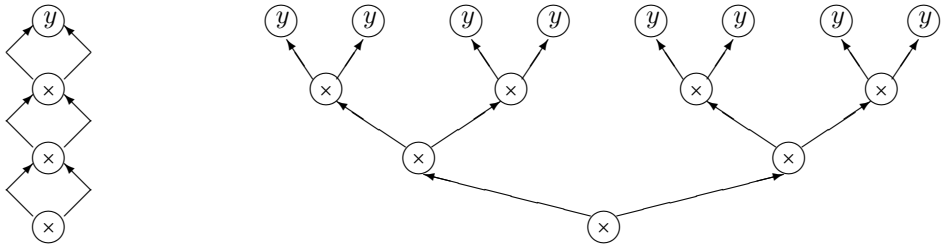
\put(50,10){{\bf Figure 1:} A graph-term and the corresponding tree-term}
\end{picture}\end{center}

By an {\bf explicit polynomial function} $\tau$ we shall mean a graph-term  not containing (at its leaves) any constants other than $0$, and not containing (at its internal nodes) any function letters other than $\successor$ (unary), $\plus$ (binary) and $\mult$ (binary).  The total number $k$ of the variables $y_1,\ldots,y_k$ occurring in (at the leaves of) $\tau$   is said to be the {\bf arity} of $\tau$. Terminologically and notationally we shall usually identify such a term $\tau$ with the $k$-ary arithmetical function represented by it under the standard arithmetical interpretation ($x\successor$ means $x\plus 1$). So, for instance, either term of Figure 1 is a unary explicit polynomial function, representing --- and identified with --- the function $y^8$.

When $\tau$ is a unary explicit polynomial function and $\cal M$ is a $\tau$ time (resp. space) machine, we  say that $\tau$ is an {\bf explicit polynomial bound} for the time (resp. space) complexity of $\cal M$. 

Another auxiliary concept that we are going to rely on in this section and later is that of a {\bf generalized HPM} ({\bf GHPM}). For a natural number $n$, an $n$-ary GHPM is defined in the same way as an HPM, with the difference that the former takes $n$ natural numbers as inputs (say, provided on a separate, read-only  {\em input tape}); such  inputs are present  at the very beginning of the work of the machine and remain unchanged throughout it. An ordinary HPM is thus nothing but a $0$-ary GHPM.  When $\cal M$ is an $n$-ary  GHPM and $c_1,\ldots,c_n$ are natural numbers, ${\cal M}(c_1,\ldots,c_n)$ denotes the HPM that works just like $\cal M$ in the scenario where the latter has received $c_1,\ldots,c_n$ as inputs. 
We will assume that some reasonable encoding (through natural numbers) of GHPMs is fixed. When $\cal M$ is a GHPM, $\code{{\cal M}}$ denotes its {\bf code}.

\begin{theorem}\label{feb9a}
%\marginpar{feb9a}
Every sequent provable in $\cltw$ has a well-behaved polynomial time and polynomial space\footnote{Note that, unlike ordinary Turing machines, not every polynomial time HPM runs in polynomial space. In this paper we leave unaddressed the natural question about whether every (interactive) computational problem with a polynomial time solution also has a polynomial space solution.}  logical solution.
Furthermore, such a solution, together with an explicit polynomial bound for both its time and space complexities, can be efficiently constructed\footnote{Here and later in similar meta-contexts, by ``efficiently'' we mean ``in polynomial time''. Also, ``can be efficiently constructed'' precisely means that there is an efficient (polynomial time) procedure that does the construction for an arbitrary proof of an arbitrary sequent.}  from a  proof of the sequent.  
\end{theorem}  

The rest of this section is exclusively devoted to a proof of the above theorem. For pedagogical reasons, we first prove the main part of the theorem, without the ``furthermore'' clause, which will be taken care only at the end of the section. Our proof of the ``pre-furthermore''  part proceeds by induction on the length of (the number of sequents involved in) a  $\cltw$-proof of a sequent $X$ and, as such, is nothing but a combination of six cases, corresponding to the six rules of $\cltw$ by which the final sequent $X$ could have been derived from its premises (if any). 

In each case, our efforts will be focused on showing how to construct an  HPM $\cal M$ --- a   logical solution of the conclusion --- from an arbitrary instance 
\[\frac{Y_1,\ldots,Y_n}{X}\] of the rule and arbitrary HPMs ${\cal N}_1,\ldots,{\cal N}_n$ ---  well-behaved polynomial time and polynomial space solutions of the premises that exist according to the induction hypothesis.  Typically it will be immediately clear from our description of $\cal M$  that  it is well-behaved and runs  in polynomial time and space,  and that its work in no way does depend on an interpretation $^*$ applied to the sequents involved, so that the solution is logical. 
Also,   our implicit assumption will be that $\cal M$'s adversary never makes illegal moves,  or else $\cal M$  easily detects the illegal behavior and retires\footnote{Technically, ``retiring'' can be understood as  going into an infinite loop that makes no moves and consumes no space.} with a decisive victory.

 Since an interpretation $^*$ is typically irrelevant in such proofs, we will usually omit it and write simply $S$ where, strictly speaking, $S^*$ is meant. That is, we identify formulas or sequents with the games into which they turn once an interpretation is applied to them. Accordingly, in  contexts  where $S^*$ has to be understood as $\ada S^*$ anyway (e.g., when talking about computability of $S^*$), we may omit ``$\ada $''  and write $S$ instead of $\ada S$.

\subsection{$\add$-Choose}\label{sep10aa} 
\[\frac{\vec{G}\ \intimpl\  F[H_i]}{\vec{G}\ \intimpl \ F[H_0\add H_1]}\]

Assume (induction hypothesis) that $\xi$ is a unary explicit polynomial function and $\cal N$ is a well-behaved $\xi$ time and $\xi$ space logical solution of the premise $\vec{G} \intimpl  F[H_i]$. We want to (show how to) construct a well-behaved logical solution $\cal M$ of the conclusion $\vec{G} \intimpl  F[H_0\add H_1]$, together with an explicit polynomial bound $\tau$ for its time and space complexities.

For the beginning, let us consider the case when the conclusion (and hence also the premise) is  closed, i.e., has no free occurrences of variables. The idea here is very simple: $\add$-Choose 
 most directly encodes an action that $\cal M$ should perform in order to successfully solve the conclusion.  Namely, $\cal M$ should choose $H_i$ and then continue playing as $\cal N$. $\cal M$ wins because the above initial move brings the conclusion down to the premise, and $\cal N$ wins the latter. And $\cal M$ automatically inherits the well-behavedness of $\cal N$. 
$\cal M$ is only ``slightly''  slower\footnote{Namely, it is simply the {\em same} in the asymptotic sense.} than $\cal N$, and it would be no problem to indicate an explicit polynomial function $\tau$ (depending on $\xi$) such that $\cal M$ runs in time $\tau$. The same applies to  space, so that we may assume that the above $\tau$ is also a polynomial bound for the space complexity of $\cal M$. Also note that our construction does not depend on an interpretation $^*$ applied to the sequents under question, so that $\cal M$ is a logical solution of the conclusion.

It now remains to consider the case when the conclusion is not closed. This   is pretty similar to the previous case. The only difference in the work of $\cal M$ will be that now, before making the move that brings the conclusion down to the premise, $\cal M$ waits till Environment chooses constants  for all  free variables of $\vec{G} \intimpl  F[H_0\add H_1]$.  Then, after choosing $H_i$,  it continues playing as $\cal N$ would play in the scenario where, at the very beginning of the play, the adversary of the latter chose the same constants for the free variables of $\vec{G} \intimpl  F[H_i]$ as $\cal M$'s environment did.

\subsection{$\adc$-Choose}\label{sep10aa2} 
This case is similar to the previous one.

\subsection{$\ade$-Choose}\label{sep10bb} 
\[\frac{\vec{G}\ \intimpl\  F[H(\mathfrak{t})]}{\vec{G}\ \intimpl \ F[\ade x H(x)]}\]
Taking into account that the choice existential quantifier is nothing but a ``long'' choice disjunction, this case is rather similar to the case of $\add$-Choose. As in that case, assume $\cal N$ is a well-behaved  logical solution of the premise and $\xi$ is an explicit polynomial bound for its time and space complexities. We want to construct a well-behaved polynomial time and polynomial space logical solution $\cal M$ of the conclusion.  

First, consider the case of $\mathfrak{t}$ being a constant.  We let ${\cal M}$ be a machine that works as follows. At the beginning,  ${\cal M}$ waits till Environment specifies some constants for all free variables of $\vec{G} \intimpl F[\ade x H(x)]$. For readability, we continue referring to the resulting game as $\vec{G} \intimpl F[\ade x H(x)]$, even though, strictly speaking, it is $e[\vec{G} \intimpl F[\ade x H(x)]]$, where $e$ is a valuation that agrees with the choices that Environment just made for the free variables of the sequent.     Now $\cal M$  makes the move that brings $\vec{G} \intimpl F[\ade x H(x)]$ down to $\vec{G} \intimpl F[H(\mathfrak{t})]$. For instance, if $\vec{G} \intimpl F[\ade x H(x)]$ is $\vec{G} \intimpl E\mlc(K\mld \ade x H(x))$ and thus $\vec{G} \intimpl F[H(\mathfrak{t})]$ is $\vec{G} \intimpl E\mlc (K\mld H(\mathfrak{t}))$, then $1.1.1.\mathfrak{t}$ is such a move.  After this move, ${\cal M}$ ``turns itself into $\cal N$'' in the same fashion as in the proof of the case of $\add$-Choose.  And, again,  ${\cal M}$ is guaranteed to be a $\tau$ time and $\tau$ space logical solution of the conclusion, where $\tau$ is an explicit polynomial function that can be easily constructed from $\xi$.  

The   case of $\mathfrak{t}$ being a variable that is among the free variables of the conclusion can be handled in a similar way, with the only difference that now, when making the move that brings the conclusion down to the premise, $\cal M$ chooses, for $x$, the constant chosen by Environment for $\mathfrak{t}$.

The remaining case is that of $\mathfrak{t}$ being a variable that is not among the free variables of the conclusion. In this case, when making the move that brings the conclusion down to the premise, $\cal M$ (arbitrarily) chooses the constant $0$ for $x$; in addition, when ``turning itself into $\cal N$'', $\cal M$ follows the scenario where $\cal N$'s adversary chose the same constant $0$ for $\mathfrak{t}$. 

In any of the above cases, it is clear that $\cal M$ is a logical solution of the conclusion, and that $\cal M$ inherits the well-behaved property of $\cal N$.

\subsection{$\ada$-Choose}\label{sep10bb2} 
This case is similar to the previous one.

\subsection{Replicate}\label{sep10bb3} 
\[\frac{\vec{G},E,\vec{K},E\intimpl F}{\vec{G},E,\vec{K}\intimpl F}\]
Remembering that we agreed to see no distinction between sequents and the games they represent, and disabbreviating $\vec{G}$,$\vec{K}$,$\intimpl$,   the premise and the conclusion of this rule can be rewritten as the following two games, respectively: 
%\marginpar{fff11a1}
\begin{equation}\label{fff11a1}
\st G_1\mlc\ldots\mlc\st G_m\ \mlc \ \st E\ \mlc \ \st K_{1}\mlc\ldots\mlc \st K_n \ \mlc\ \st E\ \ \mli\ \  F\end{equation}
%\marginpar{fff11a2}
\begin{equation}\label{fff11a2}
\st G_1\mlc\ldots\mlc\st G_m\ \mlc\ \st E\ \mlc\ \st K_{1}\mlc\ldots\mlc \st K_n\ \ \mli \ \ F\end{equation}

Assume $\cal N$ is a well-behaved  logical solution of the premise, and $\xi$ is an explicit polynomial bound for its time and space complexities. We let a  solution $\cal M$ of the conclusion be a machine that works as follows.

After Environment chooses some constants for all free variables of (\ref{fff11a2}), $\cal M$ makes a replicative move in the $\st E$ component of the latter, bringing the game down to 
%\marginpar{fff11b}
\begin{equation}\label{fff11b}
\st G_1\mlc\ldots\mlc\st G_m\ \mlc\ \st (E\circ E)\ \mlc\ \st K_{1}\mlc\ldots\mlc \st K_n\ \ \mli \ \ F \end{equation}
(more precisely, it will be not (\ref{fff11b}) but $e[(\ref{fff11b})]$, where $e$ is a valuation that agrees with Environment's choices for the free variables of the sequent.   As we did earlier, however, notationally we ignore this difference). 
  
Now we need to observe that (\ref{fff11b}) is ``essentially the same as'' (\ref{fff11a1}), so that $\cal M$ can continue playing ``essentially as'' $\cal N$ would play in the scenario where the adversary of $\cal N$ chose the same constants for free variables as the adversary of $\cal M$ just did. All that $\cal M$ needs to do to account for the minor technical differences between (\ref{fff11b}) and (\ref{fff11a1}) is to make a very simple ``reinterpretation'' of moves. Namely:
\begin{itemize}
\item Any move made within any of the $\st G_i$ or $\st K_i$ components of (\ref{fff11b}) $\cal M$ sees exactly as $\cal N$ would see the same move in the same component of (\ref{fff11a1}), and vice versa.
\item  Any replicative or focused nonreplicative move made within the left (resp. right) leaf of the $\st (E\circ E)$ component of (\ref{fff11b}) $\cal M$ sees as $\cal N$ would see the same  move as if it was made in the (single) leaf of the first (resp. second) $\st E$ component of (\ref{fff11a1}), and vice versa.
\item  Any unfocused nonreplicative move   made by Environment in the $\st (E\circ E)$ component of (\ref{fff11b}) $\cal M$ sees as $\cal N$ would see the same move made twice (but on the same clock cycle) by Environment: once in the leaf of the first $\st E$ component of 
(\ref{fff11a1}), and once in the leaf of the second $\st E$ component of 
(\ref{fff11a1}).
\end{itemize}
With a little thought, it can be seen that $\cal M$ wins because so does $\cal N$. Further, neither making the initial replicative move nor ``reinterpreting'' moves in the above fashion is expensive in terms of time or space, so that, based on  $\xi$, it would be no problem to specify an explicit polynomial bound $\tau$ 
for the time and space complexities of $\cal M$. And, as always, $\cal M$ obviously inherits the well-behavedness of $\cal N$.

\subsection{Wait}\label{iwait} 
\[\frac{Y_1,\ldots,Y_n}{X}\]
(where $n\geq 0$ and the $\adc$-, $\add$-, $\ada$-, $\ade$- and Stability conditions are satisfied). 

We shall rely on the following lemma. It   can be verified by a straightforward induction on the complexity of $Z$, which we omit. Remember that $\seq{}$ stands for the empty run.

\begin{lemma}\label{new1}
%\marginpar{new1}
For any sequent $Z$, valuation $e$ and interpretation $^*$, $\win{Z^*}{e}\emptyrun= \win{\elzi{Z}^*}{e}\emptyrun$.
\end{lemma}

Assume ${\cal N}_1,\ldots,{\cal N}_n$ are well-behaved  logical solutions of $Y_1,\ldots,Y_n$, respectively, and $\xi_1,\ldots,\xi_n$ are polynomial bounds for the time and space complexities of the corresponding machines. We  
 let ${\cal M}$, a logical solution of $X$,  be a machine that works as follows. 

 At the beginning, as always, $\cal M$ waits till Environment chooses some constants for all free variables of the conclusion. Let $e$ be a (the) valuation that agrees with the choices just made by Environment (in the previous cases, we have suppressed the $e$ parameter, but now we prefer to deal with it explicitly). So, the conclusion is now brought down to $e[X]$. After this event,  ${\cal M}$ continues waiting until Environment makes one more move.  If such a move is never made, then the run of (the $\ada$-closure of) $X$ generated in the play can be simply seen as the empty run of $e[X]$.  Due  to the Stability condition, $\elz{X}$ is classically valid, meaning that $\win{\elzi{X}}{e}\emptyrun=\pp$.  But then, in view of Lemma \ref{new1}, $\win{X}{e}\emptyrun=\pp$. This makes $\cal M$ the winner.  

 Suppose now Environment makes a move $\alpha$.    With a little thought, one can see that any (legal) move $\alpha$ by Environment brings the game $e[X]$ down to $g[Y_i]$ for a certain valuation $g$   and one of the premises $Y_i$ of the rule. For example, if $X$ is $ \intimpl (E\adc F)\mld  \ada x G(x)$, then a legal move $\alpha$ by Environment should be either $1.0.0$ or $1.0.1$ or $1.1.c$ for some constant $c$. In the case $\alpha=1.0.0$, the above-mentioned premise $Y_i$ will be $\intimpl E \mld  \ada x G(x)$, and $g$ will be the same as $e$.  In the case $\alpha=1.0.1$,   $Y_i$ will be $\intimpl F \mld  \ada x G(x)$, and $g$, again, will be the same as $e$. Finally, in the case $\alpha=1.1.c$,   $Y_i$ will be $\intimpl (E\adc F) \mld G(y)$  for a variable $y$  not occurring in $X$, and   $g$ will be the valuation that sends $y$ to the object named by $c$ and agrees with $e$ on all other variables, so that $g[\intimpl (E\adc F) \mld G(y)]$ is $e[\intimpl (E\adc F) \mld G(c)]$, with the latter being the game to which $e[X]$ is brought down by the labmove $\oo 1.1.c$. 

After the above event, ${\cal M}$ does the usual trick of turning itself into --- and continuing playing as --- ${\cal N}_i$,  with the only difference that,  if $g\not=e$, the behavior of ${\cal N}_i$ should be followed for the scenario where the adversary of the latter, at the very beginning of the play, chose constants for the free variables of $Y_i$ in accordance with $g$ rather than $e$.  

As in the preceding proofs, it can be seen that $\cal M$  is a well-behaved logical  solution of $X$. Keeping in mind that $\cal M$ is not billed for the time during which it waits for Environment to move, it is clear that, asymptotically, its time complexity does not exceed (the generously taken) 
$\xi_1\plus \ldots\plus \xi_n$. Unlike time, however, $\cal M$ {\em will} be billed for any extra space consumed while waiting. But, fortunately, it does not use any space during that period, as it simply keeps reading the leftmost blank cell of its run tape to see if Environment has made a move. So,   
 an explicit polynomial bound $\tau$ for both the time and  space complexities of $\cal M$ can be easily obtained from  $\xi_1,\ldots,\xi_n$.   

\subsection{Taking care of the ``furthermore'' clause of Theorem \ref{feb9a}.}\label{ss87}
Thus, we have shown how to construct, from a proof of $X$, an HPM $\cal M$ and an explicit  polynomial function $\tau$ such that $\cal M$ solves $X$ in time and space $\tau$. Obviously our construction is effective. It remains to see that it also is --- or, at least, can be made --- efficient. 
Of course, at every step of our construction (for each sequent  of the proof, that is), in Sections \ref{sep10aa}-\ref{iwait}, the solution $\cal M$ of the step and its time/space complexity bound $\tau$ is obtained efficiently from previously constructed $\cal M$s and $\tau$s. This, however, does not guarantee that the entire construction will be efficient as well. For instance, if the proof has $n$ steps and the size of each $\cal M$ that we construct for each step is twice the size of the previously constructed HPMs, then the size of the eventual HPM will exceed $2^n$ and thus the construction will not be efficient, even if each of the $n$ steps of it is so. 

A trick that we can use to avoid an exponential growth of the sizes of the machines that we construct and thus achieve the efficiency of the entire construction is to deal with GHPMs instead of HPMs. Namely, assume the proof of $X$ is the sequence $X_1,\ldots,X_n$ of sequents, with $X=X_n$. Let ${\cal M}_1,\ldots,{\cal M}_n$ be the HPMs constructed as we constructed ${\cal M}$s earlier in Sections \ref{sep10aa}-\ref{iwait} for the corresponding  steps/sequents. Remember that each such ${\cal M}_i$ was defined in terms of  ${\cal M}_{j_1},\ldots,{\cal M}_{j_k}$for some $j_1,\ldots,j_k< i$. For simplicity and uniformity, we may just as well say that each ${\cal M}_i$ was defined in terms of all  ${\cal M}_{1},\ldots,{\cal M}_{n}$, with those ${\cal M}_j$s that were not among ${\cal M}_{j_1},\ldots,{\cal M}_{j_k}$ simply ignored in the description of the work of ${\cal M}_i$. 
Now, for each such ${\cal M}_i$, let ${\cal M}'_i$ be the $n$-ary GHPM whose description is obtained from that of ${\cal M}_i$  by replacing  each reference to (any previously constructed) ${\cal M}_j$  by ``${\cal M}'_j(\code{{\cal M}'_1},\ldots,\code{{\cal M}'_n})$ where, for each $e\in\{1,\ldots,n\}$,  ${\cal M}'_e$ is the machine encoded by the $e$th input''.\footnote{For simplicity, here we assume that every number is a code of some $n$-ary GHPM; alternatively, ${\cal M}'_i$ can be defined so that it does nothing if any of its relevant inputs is not the code of some $n$-ary GHPM.} As it is easy to see by induction on $i$, ${\cal M}_i$ and ${\cal M}'_i(\code{{\cal M}'_1},\ldots,\code{{\cal M}'_n})$ are essentially the same, in the sense that our earlier analysis of the play and time/space complexity of the former applies to the latter just as well. So, ${\cal M}'_n(\code{{\cal M}'_1},\ldots,\code{{\cal M}'_n})$ wins $X_n$, i.e. $X$. At the same time, note that the size of each GHPM ${\cal M}'_i$ is independent of the sizes of the other (previously constructed) GHPMs. Based on this fact, with some analysis, one can see that then the HPM ${\cal M}'_n(\code{{\cal M}'_1},\ldots,\code{{\cal M}'_n})$  is indeed constructed efficiently. 

As for the explicit polynomial bounds $\tau_1,\ldots,\tau_n$ for the time and space complexities of the $n$ HPMs ${\cal M}'_1(\code{{\cal M}'_1},\ldots,\code{{\cal M}'_n})$, \ldots, ${\cal M}'_n(\code{{\cal M}'_1},\ldots,\code{{\cal M}'_n})$,  their sizes  can be easily seen to be polynomial in the size of the proof. That is because, for each $i\in\{1,\ldots,n\}$, the size of   $\tau_i$  only increases the sizes of the earlier constructed $\tau_j$s by adding (rather than multiplying by) a certain polynomial quantity.\footnote{The fact that we represent complexity bounds as graph-terms rather than tree-terms is relevant here. It would however be irrelevant if we had defined $\cltw$-proofs as trees rather than sequences of sequents. The reason why we have opted for linear rather than tree-like proofs is that, at the expense of recycling/reusing intermediate steps, linear proofs can be exponentially smaller than tree-like proofs.} Thus, the explicit bound $\tau_n$ for the time and space complexities of the eventual HPM    ${\cal M}'_n(\code{{\cal M}_1},\ldots,\code{{\cal M}_n})$ is indeed constructed efficiently. 

\section{The completeness of $\cltw$}\label{scompleteness}
%\marginpar{scompleteness}

\begin{theorem}\label{feb9b}
%\marginpar{feb9b}
Every sequent with a logical solution  is provable in $\cltw$.
\end{theorem}

\begin{proof} Assume $X$ is a sequent not provable in $\cltw$. Our goal is to show that $X$ has no logical solution (let alone a polynomial time, polynomial space and/or well-behaved logical solution).

Here we describe a {\em counterstrategy}, i.e., Environment's strategy, against which any particular HPM (in the usual role of $\pp$) loses $X^*$ for an appropriately selected interpretation $^*$. In precise terms, as a mathematical object, our counterstrategy --- let us call it $\cal C$ --- is a (not necessarily effective) function that prescribes, for each possible content of the run tape that may arise during the process of playing the game, a (possibly empty) sequence of moves that Environment should make during the corresponding clock cycle. In what follows, whenever we say that $\cal C$ wins or loses, we mean that so does $\bot$ when it acts according to such prescriptions. $\cal C$ and $\bot$ will be used interchangeably, that is.

By a {\em variables-to-constants mapping} --- or {\bf vc-mapping} for short --- for a sequent $Y$ we shall mean a function whose domain is some finite set of variables that contains all (but not necessarily only) the free variables of $Y$ and whose range is some set of constants not occurring in $Y$,  such that to any two (graphically) different variables are assigned (graphically) different constants. When $e$ is a vc-mapping for $Y$, by $e[Y]$ we shall mean the result of replacing in $Y$ each free occurrence of every variable with the constant assigned to that variable by $e$.

At the beginning of the play, $\cal C$ chooses different constants for (all) different free variables of $X$, also making sure that none of these constants are among the ones that occur in $X$. Let $g$ be the corresponding vc-mapping for $X$. This initial series of moves brings $X$ (under whatever interpretation) 
down to the constant game $g[X]$ (under the same interpretation).   
 
The way $\cal C$ works after that can be defined recursively.  At any time, $\cal C$ deals with a pair $(Y,e)$, where  $Y$ is a $\cltw$-unprovable sequent and $e$  is a vc-mapping for 
$Y$, such that $e[Y]$ is the game to which the initial $\ada X$ has been brought down ``by now''.   The  initial value of $Y$ is $X$, and the initial value of $e$ is the above vc-mapping $g$.  How $\cal C$ acts on $(Y,e)$ depends on whether $Y$ is stable or not. 

{\em CASE 1}: $Y$ is stable.  Then there should be a $\cltw$-unprovable sequent $Z$ satisfying one of the following conditions, for otherwise $Y$ would be derivable 
by Wait. $\cal C$  selects one such $Z$ (say, lexicographically the smallest one), and acts according to the corresponding prescription as given below. 

{\em Subcase 1.1:} $Y$ has the form $\vec{E}\intimpl F[G_0\adc G_1]$, and $Z$ is $\vec{E}\intimpl F[G_i]$ ($i=0$ or $i=1$). In this case, $\cal C$ makes the move that brings $Y$ down to $Z$ (more precisely, $e[Y]$ down to $e[Z]$), and calls itself on $(Z,e)$. 

{\em Subcase 1.2:} $Y$ has the form $\vec{E},F[G_0\adc G_1],\vec{K}\intimpl H$, and $Z$ is $\vec{E},F[G_i],\vec{K}\intimpl H$. This subcase is similar to the previous one. 
 
{\em Subcase 1.3:} $Y$ has the form $\vec{E}\intimpl F[\ada xG(x)]$, and $Z$ is $\vec{E}\intimpl F[G(y)]$, where $y$ is a variable not occurring in $Y$. In this case, $\cal C$ makes a move that brings $Y$ down to $\vec{E}\intimpl F[G(c)]$ for some (say, the smallest) constant $c$ such that $c$ is different from any constant occurring in $e[Y]$.  After this move, $\cal C$  calls itself on $(Z,e')$, where $e'$ is the vc-mapping for $Z$ that sends $y$ to $c$ and agrees with $e$ on all other variables. 

{\em Subcase 1.4:} $Y$ has the form $\vec{E},F[\ada xG(x)],\vec{K}\intimpl H$, and $Z$ is $\vec{E},F[G(y)],\vec{K}\intimpl H$, where $y$ is a variable not occurring in $Y$. This subcase is similar to the previous one.

$\cal C$ repeats the above until (the continuously updated) $Y$ becomes unstable. This results is some finite series of moves made by $\cal C$. We assume that all these moves are made during a single clock cycle (remember that there are no restrictions in the HPM model on how many moves Environment can make during a single cycle).   

{\em CASE 2}:  $Y$ is unstable.  ${\cal C}$ does not make any moves, but rather waits until its adversary makes a move. 

{\em Subcase 2.1}: The adversary never makes a move. Then the run of $e[Y]$ that is generated is empty.  As $Y$ is unstable, $\elz{Y}$ and hence $\elz{e[Y]}$ is not classically valid. That is, $\elz{e[Y]}$ is false in some classical model. But classical models are nothing but our interpretations restricted to elementary formulas.   So, $\elz{e[Y]}$ is false under some interpretation $^*$. This, in view of Lemma \ref{new1}, implies that $\win{(e[Y])^*}{}\seq{}=\oo$ and hence $\cal C$ is the winner.

{\em Subcase 2.2}: The adversary makes a move $\alpha$. We may assume that such a move is legal, or else $\cal C$ immediately wins.  There are two further subcases to  consider here:

{\em Subsubcase 2.2.1}: $\alpha$ is a move in the succedent, or a  nonreplicative move in one of the components of the antecedent, of $Y$. With a little thought, it can be seen that then 
$\alpha$ brings $e[Y]$ down to $e'[Z]$, where $Z$ is a sequent from which $Y$ follows by one of the four Choose rules, and $e'$ is a certain vc-mapping for $Z$. In this case, $\cal C$ calls itself on $(Z,e')$.  
 
{\em Subsubcase 2.2.2}: $\alpha$ is a replicative move in one of the components of the antecedent of $Y$. Namely, assume $Y$   (after disabbreviating $\intimpl$) is 
the game (\ref{fff11a2}) of Subsection \ref{sep10bb3}, and the replicative move is made in its $\st E$ component. This brings $e[Y]$ down to $e[(\ref{fff11b})]$. The latter, however, is ``essentially the same as'' $e[Z]$, where $Z$ abbreviates the game (\ref{fff11a1}). So, $\cal C$ can pretend that $e[Y]$ has been brought down to $e[Z]$, and call itself on $(Z,e)$. The exact meaning of ``pretend'' here is that, after calling itself on $(Z,e)$, $\cal C$ modifies its behavior --- by ``reinterpreting'' moves --- in the same style as machine $\cal M$ modified $\cal N$'s behavior in Subsection \ref{sep10bb3}.

This completes our description of the work of $\cal C $.
 
Assume a situation corresponding to Subsubcase 2.2.2 occurs only finitely many times. Note that all other cases, except Subcase 2.1, strictly  decrease the complexity of $Y$. So, the play finally stabilizes in a situation corresponding to Subcase 2.1 and, as was seen when discussing that subcase, $\cal C$ wins.  

Now, assume a situation corresponding to Subsubcase 2.2.2 occurs infinitely many times, that is, $\cal C$'s adversary makes infinitely many replications in the antecedent. And, for a contradiction, assume that
%\marginpar{ffeb12a}
\begin{equation}\label{ffeb12a}
\mbox{\em $\cal C$ loses the play of $\ada X^*$ on every interpretation $^*$.}
\end{equation}

Let $F$ be the (constant/closed) game/formula to which the succedent of the original $g[X]$ is eventually brought down. Similarly, let $\cal A$ be the set of all (closed) formulas to which various copies of various formulas of the antecedent of $g[X]$ are eventually brought down.  With a little thought and with Lemma \ref{new1} in mind, it can be seen that (\ref{ffeb12a}) implies the following:

%\marginpar{ffeb12b}
\begin{equation}\label{ffeb12b}
\mbox{\em The set $\{\elz{E}\ |\ E\in{\cal A}\}\cup\{\elz{\gneg F}\}$ is unsatisfiable (in the classical sense).}
\end{equation}

By the  compactness theorem for classical logic, (\ref{ffeb12b}) implies that, for some {\em finite} subset ${\cal A}'$ of $\cal A$, we have:

%\marginpar{ffeb12c}
\begin{equation}\label{ffeb12c}
\mbox{\em The set $\{\elz{E}\ |\ E\in{\cal A}'\}\cup\{\elz{\gneg F}\}$ is unsatisfiable (in the classical sense).}
\end{equation}

Consider a step $t$ in the work of $\cal C$ such that, beginning from $t$ and at every subsequent step, the antecedent of (the then current) $e[Y]$ contains all formulas of ${\cal A}'$. It follows easily from (\ref{ffeb12c}) that, beginning from $t$, (the continuously updated) $Y$ remains stable. This means that $\cal C$ deals only with CASE 1. But, after making a certain finite number of moves as prescribed by CASE 1, $Y$ is brought down to a stable sequent that contains no surface occurrences of $\adc,\ada$ in the succedent and no surface occurrences of $\add,\ade$ in the antecedent. Every such sequent follows from the empty set of premises by Wait, which is a contradiction because, as we know, the sequent $Y$ at any step of the work of $\cal C$ remains $\cltw$-unprovable. \end{proof}
    
\begin{remark}\label{remark}
%\marginpar{remark}
While CoL takes no interest in nonalgorithmic ``solutions'' of problems, it would still be a pity to let one fact go officially unnoticed. Virtually nothing in our proof of Theorem \ref{feb9b}  relies on the fact that an HPM whose non-existence is proven there follows an algorithmic strategy. So,
Theorem \ref{feb9b} can be strengthened by saying that, if $\cltw$ does not prove a sequent $X$, then $X$ does not even have a nonalgorithmic logical solution. Precisely defining the meaning of a ``nonalgorithmic'', or rather ``not-necessarily-algorithmic'' logical solution, is not hard. The most straightforward way to do so would be to simply take our present definition of a logical solution but generalize its underlying model of computation by allowing HPMs to have oracles  --- in the standard sense --- for whatever functions.        
\end{remark} 

As an aside, among the virtues of CoL is that it eliminates the need for many ad hoc inventions such as the just-mentioned oracles. Namely, observe that a problem $A$ is computable  by an HPM with an oracle for a function $f(x)$ if and only if the problem $\ada x\ade y\bigl(y\equals f(x)\bigr)\intimpl A$ is computable in the ordinary sense (i.e., computable by an ordinary HPM without any oracles). So, a CoL-literate person, regardless of his or her aspirations, would never really have to speak in terms of  oracles or nonalgorithmic strategies.   This explains why `{\em CoL takes no interest in nonalgorithmic ``solutions'' of problems}'.
 
\section{Logical consequence}\label{slc}
%\marginpar{slc}

The following theorem is an immediate corollary of Theorems \ref{feb9a} and \ref{feb9b}.

\begin{theorem}\label{feb9c}
%\marginpar{feb9c}
For any sequent $X$, the following conditions are equivalent: 
\begin{description}
\item[(i)] $\cltw\vdash X$.
\item[(ii)] $X$ has a logical solution.
\item[(iii)] $X$ has a well-behaved  logical solution which runs in polynomial time and space.
\end{description}
\end{theorem}

\begin{definition}\label{dlc}
%\marginpar{dlc}
Let $E_1,\ldots,E_n$ ($n\geq 0$) and $F$ be any formulas. We say that $F$ is a {\bf logical consequence} of $E_1,\ldots,E_n$ iff the sequent $E_1,\ldots,E_n\intimpl F$, in the role of $X$, satisfies any of the (equivalent) conditions (i)-(iii) of Theorem \ref{feb9c}. 
\end{definition}

As noted in Section \ref{intr}, the following rule, which we (also) call {\bf Logical Consequence},   will be the only logical rule of inference in $\cltw$-based applied systems:
\[\mbox{\em From $E_1,\ldots,E_n$ conclude $F$ as long as $F$ is a logical consequence of $E_1,\ldots,E_n$}.\]   
 
Remember from the earlier essays on CoL that, philosophically speaking, computational {\em resources} are symmetric to computational problems: what is a problem for one player to solve is a resource that the other player can use. Namely, having a problem $A$ as a computational resource intuitively means having the ability  to successfully solve/win $A$. For instance, as a resource, $\ada x\ade y(y=x^2)$ means the ability to tell the square of any number.   

According to the following thesis, logical consequence lives up to its name. A justification for it, as well as an outline of its significance, was provided in Section \ref{intr}:

\begin{thesis}\label{thesis}
%\marginpar{thesis}
Assume  $E_1,\ldots,E_n,F$ are formulas such that  there is a $^*$-independent (whatever interpretation $^*$) intuitive description and justification of a winning strategy for $F^*$, which relies on the availability and ``recyclability'' --- in the strongest sense possible --- of $E_{1}^{*},\ldots,E_{n}^{*}$ as computational resources.   Then $F$ is a logical consequence of $E_1,\ldots,E_n$. 
\end{thesis}

\begin{example}\label{intex}
%\marginpar{intex}
Imagine a $\cltw$-based applied formal theory, in which we have already proven two facts: $\cla x\bigl(x^3\equals(x\mult x)\mult x\bigr)$ (the meaning of ``cube'' in terms of multiplication) and $\ada x\ada y\ade z(z\equals x\mult y)$ (the computability of multiplication), and now we want to derive $\ada x\ade y(y\equals x^3)$ (the computability of ``cube''). This is how we can reason to justify $\ada x\ade y(y\equals x^3)$:  
\begin{quote}{\em Consider any $s$ 
(selected by Environment for $x$ in $\ada x\ade y(y\equals x^3)$). We need to find  $s^3$. Using the resource $\ada x\ada y\ade z(z\equals x\mult y)$, we first find the value $t$ of $s\mult s$, and then  the value $r$ of $t\mult s$. According to $\cla x(x^3\equals (x\mult x)\mult x)$, such an $r$ is the sought $s^3$.}
\end{quote}   

Thesis \ref{thesis} promises that the above intuitive argument will be translatable into a $\cltw$-proof of 
\[\cla x\bigl(x^3\equals(x\mult x)\mult x\bigr),\ \ada x\ada y\ade z(z\equals x\mult y)\ \intimpl \ \ada x\ade y(y\equals x^3) \]
(and hence the succedent will be derivable in the  theory by Logical Consequence as the formulas of the antecedent are already proven). Such a proof indeed exists --- see 
 Example \ref{ecube}.
\end{example}

Remember the concept of a $k$-ary explicit polynomial function $\tau$ defined in Section \ref{ssoundness}. Here we generalize it to the concept of a $(k,n)$-ary {\bf explicit polynomial functional} by allowing the term  $\tau$ to contain, on top of variables and $0,\successor,\plus,\mult$, additional $n$ ($n\geq 0$) unary function letters $f_1,\ldots,f_n$,  semantically treated as placeholders for unary arithmetical functions.  Replacing $f_1,\ldots,f_n$ by names $g_1,\ldots,g_n$ of some particular unary arithmetical functions turns $\tau$ into  the corresponding $k$-ary arithmetical function, which we shall denote by $\tau(g_1,\ldots,g_n)$.  For instance, the term of Figure 2 is a 
 $(1,2)$-ary explicit polynomial functional. Let us denote it by $\tau$. Then, if $g$ means ``square'' and $h$ means ``cube'',  $\tau(g,h)$ is  the unary arithmetical function  
 $(y^2\plus y^3)^3$. 

\begin{center} \begin{picture}(353,123)

\put(157,105){\circle{14}}
\put(155,104){$y$}

\put(144,87){\vector(1,1){11}}
\put(138,82){\circle{14}}
\put(134,80){$f_1$}

\put(170,87){\vector(-1,1){11}}
\put(176,82){\circle{14}}
\put(172,80){$f_2$}

\put(152,65){\vector(-1,1){11}}
\put(162,65){\vector(1,1){11}}

\put(157,60){\circle{14}}
\put(153,58){$\plus$}
\put(157,42){\vector(0,1){11}}
\put(157,34){\circle{14}}
\put(153,32){$f_2$}

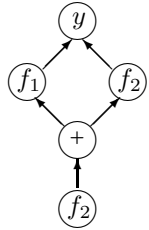
\put(0,10){{\bf Figure 2:} An explicit polynomial functional expressing $f_2\bigl(f_1(y)\plus f_2(y)\bigr)$}
\end{picture}\end{center}

The following theorem is a central result of this section. The form in which its ``furthermore'' clause is stated may seem a little strange or arbitrary (for instance, why does it deal with GHPMs rather than HPMs? And why do those GHPMs take the codes of each other as inputs?). However, the author foresees that it is exactly the present form of the theorem that will be of use when developing $\cltw$-based applied theories in the future, namely, in showing that winning strategies can be not only effectively but also efficiently extracted from proofs in such theories. Remembering the technique that we employed in Section \ref{ss87} may provide some insights into the reasons for such an expectation. 

\begin{theorem}\label{feb9d}  
%\marginpar{feb9d}
If a formula $F$ is a logical consequence of formulas $E_1,\ldots,E_n$ and $^*$ is an interpretation such that each $E_{i}^{*}$ ($1\leq i\leq n$) is computable, then  $F^*$ is computable.
Furthermore:
\begin{enumerate} 
\item There is an  efficient procedure that takes an arbitrary $\cltw$-proof of an arbitrary sequent {\em $E_1,\ldots,E_n\intimpl F$}
 and constructs a $n$-ary GHPM $\cal M$, together with a $(1,n)$-ary explicit polynomial functional $\tau$,    such that, for any interpretation $^*$, any 
$n$-ary  GHPMs ${\cal N}_1,\ldots,{\cal N}_n$   and any unary arithmetical functions $g_1,\ldots,g_n$,   if each ${\cal N}_i(\code{{\cal N}_1},\ldots,\code{{\cal N}_n})$ is a $g_{i}$ time solution of $E_{i}^{*}$, then ${\cal M}(\code{{\cal N}_1},\ldots,\code{{\cal N}_n})$ is a $\tau(g_1,\ldots,g_n)$ time solution of $F^*$. 
\item The same holds for ``space'' instead of ``time''.   
\end{enumerate}
\end{theorem}

\begin{proof}  Consider an arbitrary sequent $E_1,\ldots,E_n\intimpl F$ together with a $\cltw$-proof of it. By Theorem \ref{feb9a}, there a well-behaved logical solution  $\cal K$ of the sequent which runs in time and space $\xi$ for some explicit polynomial function $\xi$, and these $\cal K$ and $\xi$ --- fix them --- can be efficiently found. Consider an arbitrary interpretation $^*$ (which, as done before, we shall notationally suppress), 
arbitrary $n$-ary GHPMs ${\cal N}_1,\ldots,{\cal N}_n$ and assume that each ${\cal N}_i(\code{{\cal N}_1},\ldots,\code{{\cal N}_n})$ wins $E_i$ in time (resp. space) $g_i$ under that interpretation. Below we describe an $n$-ary GHPM $\cal M$ such that ${\cal M}(\code{{\cal N}_1},\ldots,\code{{\cal N}_n})$
 wins $F$ under the same interpretation. It is important to note that our {\em construction} of $\cal M$ does not depend on $^*$, ${\cal N}_1,\ldots,{\cal N}_n$,  $g_1,\ldots,g_n$ and hence on the assumptions that we have just made about them; only our {\em claim} that ${\cal M}(\code{{\cal N}_1},\ldots,\code{{\cal N}_n})$ wins $F$, and our further claims about its time and space complexities, do. 
  
To describe the above GHPM $\cal M$ means to describe the work of the HPM ${\cal M}(c_1,\ldots,c_n)$ for arbitrary numbers $c_1,\ldots,c_n$. Furthermore, we may assume that these numbers are the codes of (the earlier-mentioned, arbitrary) $n$-ary GHPMs  ${\cal N}_1,\ldots,{\cal N}_n$, because, if this is not the case (i.e., if some $c_i$ is not the code of some $n$-ary GHPM ${\cal N}_i$), how ${\cal M}(c_1,\ldots,c_n)$ works is irrelevant, so we may let $\cal M$ simply do nothing when its inputs do not have the expected forms. Thus, in what follows, we need to describe the work of the HPM ${\cal M}(\code{{\cal N}_1},\ldots,\code{{\cal N}_n})$.

As always, we let our machine  ${\cal M}(\code{{\cal N}_1},\ldots,\code{{\cal N}_n})$, at the beginning of the play, wait till Environment selects a constant for each free variable of $F$. Let us fix $e$ as the vc-mapping (see Section 
\ref{scompleteness}) whose domain is 
the set of all free variables of $E_1,\ldots,E_n\intimpl F$ such that $e$ sends every free variable of $F$ to the constant just chosen by Environment for it, and (arbitrarily) sends all other variables  to $0$. 

We describe the work of ${\cal M}(\code{{\cal N}_1},\ldots,\code{{\cal N}_n})$ afterwards at a high level. A more detailed description and analysis would be neither feasible (since it would be prohibitively long and technical) nor necessary. 

To understand the idea, let us first consider the simple case where $\cal K$ never makes any replicative moves in the antecedent of $E_1,\ldots,E_n\intimpl F$. 
  The main part of the work of ${\cal M}(\code{{\cal N}_1},\ldots,\code{{\cal N}_n})$  consists in continuously polling its run tape to see if Environment has made any new moves, combined with  simulating, in parallel, a  play of $E_1,\ldots,E_n\intimpl F$ by the machine $\cal K$ and --- for each $i\in\{1,\ldots,n\}$ --- a play of $E_i$ by the machine ${\cal N}_i(\code{{\cal N}_1},\ldots,\code{{\cal N}_n})$. 
 In this simulation, ${\cal M}(\code{{\cal N}_1},\ldots,\code{{\cal N}_n})$ ``imagines'' that, at the beginning of the play, for each free variable $x$ of the corresponding formula or sequent, 
the adversary of each machine has chosen the constant $e(x)$, where $e$ is the earlier fixed vc-mapping. After the above initial moves by the real and imaginary adversaries, each of the $n\plus 2$ games $G$ that we consider here will be brought down to $e[G]$ but, for readability and because $e$ is fixed, we shall usually omit $e$ and write simply $G$ instead of $e[G]$. 

Since we here assume that $\cal K$ never makes any replications in the antecedent of $E_1,\ldots,E_n\intimpl F$, playing this game essentially means simply playing 
%\marginpar{feb18a}
\begin{equation}\label{feb18a}
E_1\mlc \ldots\mlc E_n\mli F.
\end{equation}

We may assume that, in the real play of $F$, Environment does not make illegal moves, for then ${\cal M}(\code{{\cal N}_1},$ $\ldots,\code{{\cal N}_n})$ immediately detects this and retires, being the winner.  We can also safely assume that the simulated machines do not make illegal moves of the corresponding games, or else our assumptions about their winning those games would be wrong.\footnote{Since we  need to construct ${\cal M}(\code{{\cal N}_1},\ldots,\code{{\cal N}_n})$ no matter whether those assumptions are true or not, we can let ${\cal M}(\code{{\cal N}_1},\ldots,\code{{\cal N}_n})$ simply retire as soon as it detects some illegal behavior.}  

If so, what ${\cal M}(\code{{\cal N}_1},\ldots,\code{{\cal N}_n})$ does in the above mixture of the real and  simulated plays is that it applies copycat between $n\plus 1$ pairs of (sub)games, real or imaginary. Namely, it mimics, in (the real play of) $F$,  $\cal K$'s moves made
in the consequent of  (the imaginary play of)  (\ref{feb18a}),  and vice versa: uses   Environment's moves made in the real play of $F$ as $\cal K$'s (imaginary) adversary's moves in the consequent of (\ref{feb18a}).  Further, for each   $i\in\{1,\ldots,n\}$,  ${\cal M}(\code{{\cal N}_1},\ldots,\code{{\cal N}_n})$ uses the moves made by ${\cal N}_i(\code{{\cal N}_1},\ldots,\code{{\cal N}_n})$ in $E_i$  as $\cal K$'s adversary's 
moves in the $E_i$ component of (\ref{feb18a}), and vice versa: uses the moves made by $\cal K$ in that component as $ {\cal N}_i(\code{{\cal N}_1},\ldots,\code{{\cal N}_n})$'s adversary's  moves in $E_i$.

Therefore, the final positions   hit by the $n\plus 2$ imaginary and real plays 
\[\mbox{$E_1,\ \ldots,\ E_n,\ \ E_1\mlc\ldots\mlc E_n\mli F$ \ and \ $F$}\]
will retain the above forms, i.e., will be
\[\mbox{$E'_1,\ \ldots,\ E'_n,\ \ E'_1\mlc\ldots\mlc E'_n\mli F'$ \ and \ $F'$}\]
for some $E'_1,\ldots,E'_n,F'$. Our assumption that the machines ${\cal N}_1 (\code{{\cal N}_1},\ldots,\code{{\cal N}_n}),\ldots,{\cal N}_n (\code{{\cal N}_1},\ldots,\code{{\cal N}_n})$ and ${\cal K}$ win the games $E_1, \ldots, E_n$ and $F_1\mlc\ldots\mlc F_n\mli F$ implies that each $G\in\{E'_1,\ \ldots,\ E'_n,\ E'_1\mlc\ldots\mlc E'_n\mli F'\}$ is $\pp$-won, in the sense that $\win{G}{}\seq{}=\pp$. It is then obvious that so should be  $F'$. Thus, the (real) play of $F$ brings it down to the $\pp$-won  $F'$, meaning that ${\cal M}(\code{{\cal N}_1},\ldots,\code{{\cal N}_n})$ wins $F$. Note that   ${\cal M}$ can be constructed efficiently, as promised in the theorem.
 
The next thing to clarify is why a bound for the running time (resp. space) of ${\cal M}(\code{{\cal N}_1},\ldots,\code{{\cal N}_n})$ can be expressed as $\tau(g_1,\ldots,g_n)$, where  $\tau$ is an $n$-ary explicit polynomial functional.  Let $\phi(x)$ be an abbreviation of
$g_1(x)\plus\ldots\plus g_n(x)\plus \xi(x)$. Thus, each of the $n\plus 1$ simulated machines that we consider runs in --- the very generously selected for the sake of simplicity --- time (resp. space) $\phi$. Note that, because $\phi$ contains $\xi$, we have $\phi(x)\geq x$.  Next, let us fix $\mathfrak{b}$ as twice the sum of the maximum lengths of legal runs of 
$E_1,\ldots,E_n$ and 
 $F$. 
We may assume that, in any case, $\mathfrak{b}\geq n\plus 1$. 

The simulation and copycat performed by ${\cal M}(\code{{\cal N}_1},\ldots,\code{{\cal N}_n})$ do impose some time and space overhead. But the latter is only polynomial and, in our subsequent analysis, can be safely ignored. That is, for the sake of simplicity, we are going to pretend that the space that ${\cal M}(\code{{\cal N}_1},\ldots,\code{{\cal N}_n})$ consumes does not exceed the sum of the spaces consumed by the simulated machines, that ${\cal M}(\code{{\cal N}_1},\ldots,\code{{\cal N}_n})$ copies moves in its copycat routine instantaneously, and that the times that ${\cal M}(\code{{\cal N}_1},\ldots,\code{{\cal N}_n})$ ever spends ``thinking'' about what move to make are the times during which it is waiting for simulated machines to make one or several moves. Furthermore, we will pretend that simulation happens in a truly parallel fashion, in the sense that ${\cal M}(\code{{\cal N}_1},\ldots,\code{{\cal N}_n})$ spends a single clock cycle on tracing a single computation step of  all machines simultaneously.  Also, a move $\alpha$ made in $F$, when ``copied'' in the consequent of (\ref{feb18a}),   will become $1.\alpha$; however, we shall ignore this minor difference and pretend that the size of $\alpha$ is the same as that of $1.\alpha$. Similarly for moves made in $E_1,\ldots,E_n$ and ``copied'' in the antecedent of (\ref{feb18a}).

We start with space complexity. Consider an arbitrary play (computation branch) of ${\cal M}(\code{{\cal N}_1},\ldots,\code{{\cal N}_n})$, and an arbitrary clock cycle $c$. Let $\ell$ be the background of $c$.  In the simulations of $\cal K$ and ${\cal N}_1(\code{{\cal N}_1},\ldots,\code{{\cal N}_n})$, \ldots, ${\cal N}_n (\code{{\cal N}_1},\ldots,\code{{\cal N}_n})$, every move made by the imaginary adversary of one of these machines is a copy of either a move made by Environment in the real play, or a move made by one of the machines ${\cal K}$, ${\cal N}_1(\code{{\cal N}_1},\ldots,\code{{\cal N}_n})$, \ldots, ${\cal N}_n(\code{{\cal N}_1},\ldots,\code{{\cal N}_n})$ during simulation.  
Let $\beta_1,\ldots,\beta_m$ be the moves by simulated machines that ${\cal M}(\code{{\cal N}_1},\ldots,\code{{\cal N}_n})$ detects by time $c$, arranged according to the times of their detections.\footnote{According to our simplified view, ${\cal M}(\code{{\cal N}_1},\ldots,\code{{\cal N}_n})$ detects such a move at the same time (clock cycle) as the time at which the move is made in the simulated play by the corresponding machine. So, in case two or more of the moves $\beta_1,\ldots,\beta_m$ are made simultaneously by the corresponding machines, there can be more than one arrangement of these moves ``according to their detection times''; which one is chosen, however, is irrelevant.} Let  \({\cal H}_1,\ldots,{\cal H}_m\in\{{\cal K},{\cal N}_1 (\code{{\cal N}_1},\ldots,\code{{\cal N}_n}),\ldots,{\cal N}_n (\code{{\cal N}_1},\ldots,\code{{\cal N}_n})\}\) be the machines that made these moves, respectively.      The size of $\beta_1$ cannot exceed $\phi(\ell)$. That is because, by the time when    ${\cal H}_1$ made the move $\beta_1$,  all (if any) moves by ${\cal H}_1$'s imaginary adversary were copies of moves made by Environment in the real play rather than moves made by some other simulated machines, and hence the background of $\beta_1$ in the simulated play of ${\cal H}_1$ did not exceed $\ell$. And this means that the $\phi$ space machine ${\cal H}_1$ would not have enough space to construct $\beta_1$ on its work tape before making this move if the size of the latter was greater than $\phi(\ell)$. For similar reasons, with $\phi(\ell)$ now acting in the role of $\ell$, the size of $\beta_2$ cannot exceed $\phi(\phi(\ell))$. Similarly, the size of $\beta_3$ cannot exceed $\phi(\phi(\phi(\ell)))$, etc. Also notice that at most $\mathfrak{b}$ moves can be made altogether in the mixture of the real and the imaginary plays, so that $m\leq \mathfrak{b}$. Thus, the size of no move made in this mixture by time $c$ exceeds $\phi^\mathfrak{b}(\ell)$ ($\phi^\mathfrak{b}$ means the $\mathfrak{b}$-fold composition of $\phi$ with itself). Therefore, as all simulated machines run in space $\phi$, the space consumed by each of the $n\plus 1$ simulations by time $c$ does not exceed  $\phi^{\mathfrak{b}+1}(\ell)$. Hence the total space that ${\cal M}(\code{{\cal N}_1},\ldots,\code{{\cal N}_n})$ has used by time $c$ does not exceed $(n\plus 1)\mult\phi^{\mathfrak{b}+1}(\ell)$. Let us be generous and, remembering that $\mathfrak{b}\geq n\plus 1$,  write the latter as $\mathfrak{b}\mult\phi^{\mathfrak{b}+1}(\ell)$ instead, to eliminate explicit dependence on $n$ (this is helpful for our later purposes).   Of course, $\mathfrak{b}\mult\phi^{\mathfrak{b}+1}(\ell)$, even after accounting for various overheads that we have suppressed in our simplified bookkeeping, can be written as   $\tau(g_1,\ldots,g_n)$, where $\tau$ is an $n$-ary explicit polynomial functional.  This is exactly the sought bound for the space complexity of ${\cal M}(\code{{\cal N}_1},\ldots,\code{{\cal N}_n})$. Obviously $\tau$ can be constructed efficiently.

Now we look at time complexity. Consider an arbitrary play of ${\cal M}(\code{{\cal N}_1},\ldots,\code{{\cal N}_n})$, and an arbitrary clock cycle $c$ on which ${\cal M}(\code{{\cal N}_1},\ldots,\code{{\cal N}_n})$ makes a move $\alpha$. Let $\ell$, $\beta_1,\ldots,\beta_m$, ${\cal H}_1,\ldots,{\cal H}_m$ be as in the previous case.   For reasons similar to those employed there, we find that, for each $i\in\{1,\ldots,m\}$, the size of $\beta_i$ does not exceed  $\Re$, where $\Re=\phi^\mathfrak{b}(\ell)$. 
Let $t_1,\ldots, t_k$ be the times at which the above moves $\beta_1,\ldots,\beta_m$ were detected by ${\cal M}(\code{{\cal N}_1},\ldots,\code{{\cal N}_n})$ (that is, made by the corresponding machines).   
Further, let $k$ be the timecost of $\alpha$, and let   $d=c\minus k$. Let $j$ be the smallest integer among $1,\ldots,m$ such that $t_j\geq d$. Since ${\cal H}_j$ runs in time $\phi$, it is clear that $t_j\minus d$ does not exceed $\phi(\Re)$. Nor does $t_{i\plus 1}\minus t_{i}$ for any $i\in\{j,\ldots,m\}$. Hence $t_m\minus d\leq  (m\minus j\plus 1)\mult \phi(\Re)$. But notice that $\beta_m$ is a move made by $\cal K$ in the consequent of (\ref{feb18a}), immediately (by our simplifying assumptions) copied by ${\cal M}(\code{{\cal N}_1},\ldots,\code{{\cal N}_n})$ in the real play when it made its move $\alpha$. In other words, $c=t_m$. And $c\minus d=k$. So, $k$ --- the timecost of $\alpha$ --- does not exceed $(m\minus j\plus  1)\mult  \phi(\Re)$ and hence $\mathfrak{b}\mult \phi(\Re)$.  Nor does the size of $\alpha$.  Now $\mathfrak{b}\mult \phi(\Re)$, even after accounting for various overheads that we have suppressed in our simplified bookkeeping, can be written as   $\tau(g_1,\ldots,g_n)$, where $\tau$ is an $n$-ary explicit polynomial functional, which can be constructed efficiently.  This is exactly the sought bound for the time complexity of ${\cal M}(\code{{\cal N}_1},\ldots,\code{{\cal N}_n})$.

Whatever we have said so far was about the simple case when $\cal K$ makes no replicative moves in the antecedent of $E_1,\ldots,E_n\intimpl F$. How different is the general case, where $\cal K$ can make replications? Not very different. The overall work of ${\cal M}(\code{{\cal N}_1},\ldots,\code{{\cal N}_n})$ remains the same, with the only difference that, every time $\cal K$ replicates one of $E_i$ (more precisely, to whatever a given copy of $E_i$ has evolved by that time), ${\cal M}(\code{{\cal N}_1},\ldots,\code{{\cal N}_n})$ splits the corresponding simulation of ${\cal N}_i(\code{{\cal N}_1},\ldots,\code{{\cal N}_n})$ into two identical copies, with the same past but possibly diverging futures. This increases the number of simulated plays and the corresponding number of to-be-synchronized (by the copycat routine) pairs of games  by one, but otherwise ${\cal M}(\code{{\cal N}_1},\ldots,\code{{\cal N}_n})$ continues working as in the earlier described scenario. ${\cal M}(\code{{\cal N}_1},\ldots,\code{{\cal N}_n})$ is guaranteed to win for the same  reasons as before. Furthermore, the time and space complexity analysis that we provided earlier still remains valid. The point is that, as we remember, $\cal K$ is well-behaved, so that, even if it makes replicative moves, it does so only a certain bounded number of times. Thus, the parameter $\mathfrak{b}$ on which we relied earlier still remains constant (the new $\mathfrak{b}$ is only by a constant factor greater than the old one), and thus so do the  the overall number of simulations performed by ${\cal M}(\code{{\cal N}_1},\ldots,\code{{\cal N}_n})$ and 
the   depth of compositions of $\phi$ within the $\tau$ term.  
\end{proof}

The following fact is an immediate corollary of Theorem  \ref{feb9d}:

\begin{corollary}\label{feb9e} \ 
%\marginpar{feb9e}

1. Whenever a formula $F$ is a logical consequence of formulas $E_1,\ldots,E_n$ and the latter have polynomial time solutions under a given interpretation $^*$,  so does the former. Such a solution, together with an explicit polynomial bound for its time complexity, can be  efficiently constructed from a $\cltw$-proof of {\em $E_1,\ldots,E_n\intimpl F$}, solutions for $E_{1}^{*},\ldots,E_{n}^{*}$, and explicit polynomial bounds for their time complexities.

2. The same holds for ``space'' instead of ``time''.   
\end{corollary}

But the import of Theorem \ref{feb9d} extends far beyond the above corollary. The theorem implies that, for any class $\Omega$ of functions that contains all polynomial functions and is closed under composition, the rule of Logical Consequence preserves $\Omega$-time and $\Omega$-space computabilities. This means that $\cltw$ is an adequate logical basis for a wide class of complexity-oriented or complexity-sensitive applied systems. Among those, other than systems for  
 polynomial time and polynomial space computabilities, are many other naturally emerging theories worth studying, such as those for 
 elementary recursive (in the sense of Kalmar) computability,\footnote{It can be seen that elementary  recursive time is equivalent to elementary recursive space, so we omit the specification ``time'' or ``space'' here. The same applies to primitive recursive computability, provably recursive computability and general recursive computability.} primitive recursive computability, ({\bf PA}-) provably recursive computability, general recursive computability, etc. $\cltw$ --- more precisely, the associated rule of Logical Consequence --- is adequate because, on one hand, by Theorem \ref{feb9d}, it is sound for all such systems, and, on the other hand, by Theorem \ref{feb9b}  and/or Thesis \ref{thesis} (feel free to also throw  Remark \ref{remark} into the mix), it is as strong as a logical rule of inference could possibly be. 
  
\section{Some admissible rules of $\cltw$}
We say that a given rule is {\bf admissible} in $\cltw$ iff, for every instance of the rule, whenever all premises are provable, so is the conclusion. 

Before closing this paper, we want to identify a few admissible rules of $\cltw$ for possible future use and reference. 
 Among such rules are:

\begin{center}
\begin{picture}(302,60)

\put(15,45){\bf Exchange}
\put(0,22){$\frac{\mbox{$\vec{E},H,G,\vec{K}\intimpl F$}}{\mbox{$\vec{E},G,H,\vec{K}\intimpl F$}}$}

\put(114,45){\bf Weakening}
\put(117,22){$\frac{\mbox{$\vec{E}\intimpl F$}}{\mbox{$\vec{E},\vec{K}\intimpl F$}}$}

\put(242,45){\bf Cut}
\put(200,22){$\frac{\mbox{$\vec{E}\intimpl F\hspace{30pt}\vec{K},F\intimpl G$}}{\mbox{$\vec{E},\vec{K}\intimpl G$}}$}

\end{picture}
\end{center}

\begin{fact}\label{feb19a}
%\marginpar{feb19a}
Exchange, Weakening and Cut are admissible in $\cltw$.
\end{fact}

\begin{proof} Section 6 of \cite{Japtowards} proves that these rules preserve a certain concept of validity. Section 8 of the same paper also proves that $\cltw$ is sound and complete with respect to that concept of validity. So, the rules are admissible in $\cltw$. Of course, the admissibility of Exchange and Weakening can as well be seen directly, using a straightforward syntactic argument.
\end{proof}

Note that, in view of the admissibility of Exchange and Weakening (and the presence of Replicate), an equivalent formulation of $\cltw$ would be one that sees the antecedent of a sequent as a set rather than sequence or even a multiset of formulas. Such a formulation would only have five rules, with the  rule of Replicate  being trivial and hence redundant there.  

Among the six rules of $\cltw$, the expensive one is Wait as, at times, it may require too many premises. Using the following four rules instead of Wait can very significantly shorten proofs. The notation $F^{\vee}[G]$ (resp. $F^{\wedge}[G]$) employed in our formulation of those rules means the same as our earlier (Section \ref{ss8}) agreed-on $F[G]$, with the additional restriction that the fixed surface occurrence of $G$ is not in the scope of $\mlc$ (resp. $\mld$). Also,  $y$ is a variable not occurring in the conclusion.  

\begin{center}
\begin{picture}(382,60)

\put(34,45){\bf $\adc$-Introduction}
\put(0,22){$\frac{\mbox{$\vec{E}\intimpl F^{\vee}[G_0]\hspace{35pt}\vec{E}\intimpl F^{\vee}[G_1]$}}{\mbox{$\vec{E}\intimpl F^{\vee}[G_0\adc G_1]$}}$}

\put(252,45){\bf $\add$-Introduction}
\put(190,22){$\frac{\mbox{$\vec{E},F^{\wedge}[G_0],\vec{K}\intimpl H\hspace{35pt}\vec{E},F^{\wedge}[G_1],\vec{K}\intimpl H$}}{\mbox{$\vec{E},F^{\wedge}[G_0\add G_1],\vec{K}\intimpl H$}}$}
\end{picture}
\end{center}

\begin{center}
\begin{picture}(382,60)

\put(34,45){\bf $\ada$-Introduction}
\put(32,22){$\frac{\mbox{$\vec{E}\intimpl F^{\vee}[G(y)]$}}{\mbox{$\vec{E}\intimpl F^{\vee}[\ada xG(x)]$}}$}

\put(252,45){\bf $\ade$-Introduction}
\put(237,22){$\frac{\mbox{$\vec{E},F^{\wedge}[G(y)],\vec{K}\intimpl H$}}{\mbox{$\vec{E},F^{\wedge}[\ade xG(x)],\vec{K}\intimpl H$}}$}
\end{picture}
\end{center}

\begin{fact}\label{feb19b}
%\marginpar{feb19b}
$\adc$-Introduction, $\add$-Introduction, $\ada$-Introduction and $\ade$-Introduction are admissible in $\cltw$.
\end{fact}

\begin{proof} Easy induction, details of which we omit.  
\end{proof}

\end{document}